\documentclass[11pt]{article}
\pdfoutput=1
\usepackage{soul}
\usepackage{jheppub}
\usepackage{amsfonts}
\usepackage{amsmath}
\usepackage{amsthm}
\allowdisplaybreaks[4]         
\usepackage{amssymb}   
\usepackage{euscript}   
\usepackage{cleveref}    
\usepackage[dvipsnames]{xcolor}          
\usepackage{tensor}     
\usepackage{graphicx}
\usepackage{pifont}
\usepackage{caption}
\usepackage{subcaption}   
\usepackage{hyperref}
\hypersetup{
	colorlinks =NavyBlue,
	linkcolor =NavyBlue,
	citecolor=MidnightBlue,
	filecolor=NavyBlue,
	urlcolor=NavyBlue,
}

\usepackage{tikz}
\usetikzlibrary{arrows.meta}

\newcommand{\req}[1]{(\ref{#1})} 

\newcommand{\bea}{\begin{eqnarray}}
\newcommand{\eea}{\end{eqnarray}}
\newcommand{\ba}{\begin{eqnarray}}
\usepackage{braket}
\newcommand{\ea}{\end{eqnarray}}

\newcommand{\beq}{\begin{equation}}
\newcommand{\eeq}{\end{equation} }
\newcommand{\beqa}{\begin{eqnarray}}

\newcommand{\eeqa}{\end{eqnarray}}
\newcommand{\beqar}{\begin{eqnarray*}}
\newcommand{\eeqar}{\end{eqnarray*}}

\newcommand{\be}{\begin{equation}}
\newcommand{\ee}{\end{equation}}
\newcommand{\diff}{\mathrm{d}}

\renewcommand{\req}[1]{(\ref{#1})}

\newcommand{\eg}{{\it e.g.,}\ }
\newcommand{\ie}{{\it i.e.,}\ }

\newcommand{\Rho}{\mathrm{P}}

\newtheorem{theorem}{Theorem}

\newtheorem{defi}{Definition}
\newtheorem{notion}{Notion}

\newtheorem{proposition}{Proposition}

\newtheorem{remark}{Remark}

\usepackage{wasysym}
\usepackage{stmaryrd}

\definecolor{shadecolor}{rgb}{.25,.25,.25}

\title{ \boldmath Birkhoff implies Quasi-topological}

\author[\lightning]{Pablo Bueno,}
\author[\scriptsize \sun]{Robie A. Hennigar,}
\author[\lightning]{\'Angel J. Murcia}

\affiliation[\lightning]{Departament de Física Quàntica i Astrofísica, Institut de Ciències del Cosmos Universitat de Barcelona, Martí i Franquès 1, E-08028 Barcelona, Spain \vspace{0.1cm}}
\affiliation[\scriptsize \sun]{Centre for Particle Theory, Department of Mathematical Sciences, Durham University, Durham DH1 3LE, UK \vspace{0.1cm}}

\emailAdd{pablobueno@ub.edu}
\emailAdd{robie.a.hennigar@durham.ac.uk}
\emailAdd{angelmurcia@icc.ub.edu}

\date{\today}
\abstract{Quasi-topological gravities (QTGs) are higher-curvature extensions of Einstein gravity in $D\geq 5$ spacetime dimensions. 
Throughout the years, different notions of QTGs constructed from analytic functions of polynomial curvature invariants have been introduced in the literature. In this paper, we show that all such definitions may be reduced to three distinct inequivalent notions: type I QTGs, for which the field equations evaluated on a single-function static and spherically symmetric ansatz are second order; type II QTGs, whose field equations on general static and spherically symmetric backgrounds are second order; and type III QTGs, for which the trace of the field equations on a general background is second order. We show that type II QTGs are a subset of type I QTGs and that type III QTGs are a subset of type II QTGs modulo pure Weyl invariants. Moreover, we prove that type II QTGs possess second-order equations on general spherical backgrounds. This allows us to prove that any theory satisfying a Birkhoff theorem is a type II QTG, and that the reverse implication also holds up to a zero-measure set of theories. For every theory satisfying Birkhoff's theorem, the most general spherically symmetric solution is a generalization of the Schwarzschild spacetime characterized by a single function which satisfies an algebraic equation.
 }

\begin{document} 
\vspace*{2cm} 
\maketitle
\flushbottom


\section{Introduction and summary of results}

The study of gravitational actions involving higher powers of the curvature has a long history and broad motivations. From an effective field theory perspective, terms of this nature arise as unavoidable high-energy corrections to general relativity (GR)~\cite{Donoghue:1994dn, Endlich:2017tqa}. Higher-curvature terms can improve the renormalizability of gravitational theory, though at the expense of new (and potentially problematic) degrees of freedom~\cite{tHooft:1974toh, Stelle:1976gc, Stelle:1977ry,Goroff:1985th}. Specific combinations of higher-curvature terms are predicted by different approaches to quantum gravity. This is the case, for example, in string theory --- though only the first few of these ``$\alpha'$ corrections'' are known explicitly~\cite{Fradkin:1984pq, Callan:1985ia,Fradkin:1985ys,Zwiebach:1985uq, Gross:1986mw, Metsaev:1987zx}. Higher-curvature modifications to gravity also play a role in cosmology, most famously in the model of Starobinsky inflation~\cite{Starobinsky:1980te}. Beyond these specific realizations, higher-curvature theories provide a useful laboratory for probing universality: when a statement holds across broad classes of deformations of GR, it points to a model-independent lesson~\cite{Wald:1993nt,Myers:2010xs, Myers:2010tj,Mezei:2014zla, Bueno:2015rda,Bueno:2015xda,Miao:2015dua,Bianchi:2016xvf, Bobev:2017asb, Bueno:2018yzo,Bueno:2022jbl}.

Among the various possible higher-curvature theories, one class that has attracted considerable attention is \textit{quasi-topological gravities} (QTGs). The first examples of QTGs were found in 2010 almost simultaneously by two groups for two different purposes. On the one hand, a cubic-in-curvature interaction was identified in~\cite{Oliva:2010eb} while searching for beyond-GR theories that satisfy a Birkhoff theorem {--- which is not generically fulfilled in generic higher-curvature theories, see e.g. \cite{Pechlaner1966,Stelle:1977ry,Clifton:2006ug,Faraoni:2010rt,Capozziello:2011wg,Ghosh:2024tlk}}. On the other hand, an equivalent cubic-in-curvature interaction was identified in~\cite{Quasi} during a search for a holographic toy model for four-dimensional conformal field theories with distinct central charges $a \neq c$. After these initial discoveries, new examples of QTGs were identified first at quartic order in curvature~\cite{Dehghani:2011vu}, then quintic order~\cite{Cisterna:2017umf}, and finally at arbitrary order in the spacetime curvature and in all dimensions $D \ge 5$~\cite{Bueno:2019ycr, Bueno:2022res, Moreno:2023rfl,Moreno:2023arp}.\footnote{Early work on QTGs focused on \emph{polynomial} curvature invariants, and we adopt that focus here. Recent results show that allowing non-polynomial invariants yields theories that retain QTG properties and exist in four dimensions~\cite{Bueno:2025zaj} --- see also \cite{Chinaglia:2017wim, Colleaux:2017ibe,Colleaux:2019ckh}. In fact, such theories all have equations of motion that can be constrained by `lifting' certain two-dimensional Horndeski theories to higher-dimensions~\cite{Carballo-Rubio:2025ntd}.}

Roughly speaking, QTGs are defined by their behaviour in spherical symmetry: they possess static and  spherically symmetric black hole solutions determined by a single metric function, $-g_{tt}=g_{rr}^{-1}\equiv f(r)$, which satisfies an algebraic equation in terms of its ADM mass~\cite{Oliva:2010eb,Quasi}. The algebraic equation is formally equivalent to that in Lovelock gravity~\cite{Lovelock1, Lovelock2, Wheeler:1985nh}. However, while Lovelock theory of order $N$ in curvature is nontrivial only for $D \ge 2 N + 1$, QTGs of any order exist in any dimension $D \ge 5$. This simplicity motivated much of the early work on the theories, where they were studied as gravitational and holographic toy models to probe beyond general relativity~\cite{Oliva:2010eb,Quasi, Dehghani:2011vu,Myers:2010jv,Kuang:2010jc,Brenna:2011gp,Oliva:2011xu,Kuang:2011dy,Dehghani:2011hm, Brenna:2012gp,Bazrafshan:2012rn,Dehghani:2013ldu, Hennigar:2015esa,Chernicoff:2016qrc,Cisterna:2017umf,Cisterna:2018tgx,Fierro:2020wps,Bueno:2020odt}. It was later realized that, despite the careful tuning of coupling constants required to realize the structural properties of QTGs, the theories are broad enough to capture general corrections to vacuum gravity within an effective field theory framework~\cite{Bueno:2019ltp, Bueno:2019ycr}. Most recently, it has been found that a complete resummation of QTGs --- including the effects of terms at all orders in curvature --- generically produces regular black hole solutions~\cite{Bueno:2024dgm, Bueno:2024eig, Bueno:2024zsx,Bueno:2025gjg}. This has attracted considerable attention as the first explicit model in which dynamical aspects of regular black holes can be probed~\cite{Konoplya:2024hfg, DiFilippo:2024mwm, Konoplya:2024kih, Ma:2024olw, Cisterna:2024ksz, Ditta:2024iky, Frolov:2024hhe,Wang:2024zlq,Bueno:2025dqk, Fernandes:2025fnz, Fernandes:2025eoc,Hennigar:2025ftm,Aguayo:2025xfi,Cisterna:2025vxk,Boyanov:2025pes,Konoplya:2025uta,Fan:2025jow, Eichhorn:2025pgy, Ling:2025ncw,Arbelaez:2025gwj,Xie:2025auj,Chen:2025pgg}.

Since their introduction, QTGs have been formulated in several ---  sometimes inequivalent --- ways. Some definitions are tailored to static, spherically symmetric metrics; others drop the assumption of staticity; still others impose requirements meant to hold beyond spherical symmetry. Which of these formulations are equivalent, and which are genuinely broader, remains unclear. This paper clarifies these relationships --- see Fig.\,\ref{fig:qtg-nested} below for a schematic summary.  In addition, we prove a general result which holds for arbitrary higher-curvature theories built from general (polynomial) contractions of the Riemann tensor and the metric: namely, that any of such theories which satisfies a Birkhoff theorem --- in a precise sense that is explained in the document ---  necessarily belongs to one of the QTG classes characterized in this paper. This goes beyond previous results in the literature, which proved a Birkhoff theorem either for Lovelock gravities \cite{Charmousis:2002rc,Zegers:2005vx} or for certain classes of QTGs \cite{Oliva:2011xu,Oliva:2012zs,Cisterna:2017umf,Bueno:2024dgm,Bueno:2024qhh,Bueno:2024zsx}.

For the benefit of the reader, after introducing the various (six) notions of QTGs coined in the literature throughout the years, we opt to begin our document by presenting the main results. This is the content of the following two subsections. The derivation and proofs of our results are presented in subsequent sections, which are structured as follows. In Section \ref{sec2} we identify the first and broadest set of QTGs, namely, the one involving theories for which the equations of motion on a single-function static and spherically symmetric ansatz are second order. We prove various propositions for this type of theories and show that two of the previously considered notions of QTG are equivalent to this one. In Section \ref{sec3} we characterize the second type of QTGs, corresponding to theories for which the equations of motion on general static and spherically symmetric metrics are second order. This type is contained in the first class of QTGs and we show that three of the six previously considered notions of QTG are all equivalent to one another. In Section \ref{sec4} we prove several results concerning the structure of the equations of higher-curvature theories whose equations of motion on spherically symmetric backgrounds are second order. Specifically, we show that QTGs of the second type necessarily possess second-order equations on general spherically symmetric backgrounds.  
In Section \ref{sec5} we prove that any theory that satisfies a Birkhoff theorem is a QTG of the second type, showing that the reverse implication also holds except for a zero measure set of theories whose static and spherically symmetric solutions are characterized by more than one free continuous parameter. In Section \ref{sec6} we characterize the last class of QTGs, namely, those for which the traced field equations are second order in derivatives on general backgrounds. We show that any such QTG can be written as a QTG of the second type plus an analytic function of two invariants built, respectively, from contractions of two and three Weyl tensors. Finally, we conclude with some future directions and an appendix with an alternative characterization for QTGs of the second type.

{\bf Note:} Throughout the document, we will be considering theories of gravity $\mathcal{L}(g^{ab},R_{cdef})$ that may be written as:
\begin{equation}
    \mathcal{L}(g^{ab},R_{cdef})=\sum_{n=0}^\infty \sum_{j=1}^{g_n} \alpha_j^{(n)} \mathcal{R}_{j}^{(n)}\,,
    \label{eq:anapol}
\end{equation}
where $\mathcal{R}_{j}^{(n)}$ is a curvature invariant of order $n$, $g_n$ stands for the number of independent curvature invariants at order $n$, each of them labeled with the index $j$, and $\alpha_j^{(n)}$ are unspecified couplings with units of length$^{2(n-1)}$ --- one could set all $\alpha_j^{(n)}=0$ for $n \geq n_{\rm max}$, thus obtaining a polynomial higher-curvature theory. Whenever a theory has the form \eqref{eq:anapol}, we will just say that the theory is \emph{written as an analytic function of polynomial curvature invariants}, or synonyms thereof. In fact, any theory in the present manuscript, unless otherwise stated, may be assumed to adopt the form \eqref{eq:anapol}.

\subsection{The various notions of QTG}

Let $\mathcal{L}=\mathcal{L}(g^{ab}, R_{cdef})$ be a $D$-dimensional theory of gravity (with $D \geq 5$) formed by combinations of arbitrary polynomial contractions of Riemann curvature tensors with the metric --- perhaps featuring an infinite tower of higher-curvature corrections. Consider a general\footnote{Throughout this manuscript, we will consider various spherically symmetric backgrounds in which the warp factor multiplying the metric of the $(D-2)$-dimensional round sphere is given precisely by the radial coordinate. In general, spherical symmetry does not determine such a warp factor and fixing it to be precisely the radial coordinate is allowed as long as the gradient of the warp factor is (asymptotically) space-like.} static and spherically symmetric ansatz (SSS) for the metric:
\begin{equation}
\mathrm{d}s_{N,f}^2=-N(r)^2 f(r) \mathrm{d} t^2+\frac{\mathrm{d} r^2}{f(r)}+r^2 \mathrm{d} \Omega_{D-2}^2\,.
\label{eq:sssintro}
\end{equation}
We will also make use of the single-function SSS ansatz, which entails considering \req{eq:sssintro} with $N(r)=1$, namely,
\begin{equation}
\mathrm{d}s_{f}^2=-f(r) \mathrm{d} t^2+\frac{\mathrm{d} r^2}{f(r)}+r^2 \mathrm{d} \Omega_{D-2}^2\,.
\label{eq:sssn}
\end{equation}
Define $L_{N,f}=r^{D-2} N(r) \mathcal{L} \vert_{N,f}$ (resp., $L_{f}=r^{D-2} \mathcal{L} \vert_{f}$) to be the evaluation of the theory on the general SSS ansatz \eqref{eq:sssintro} (resp., on the single-function SSS ansatz \eqref{eq:sssn}). From now on, every tensor marked by $ \vert_{N,f}$ (resp., $\vert_{f}$) is assumed to be evaluated on \eqref{eq:sssintro} (resp., \eqref{eq:sssn}). We will sometimes use $\vert_{N=1}$ to specify that a tensor evaluated on \eqref{eq:sssintro} is evaluated on the single-function SSS ansatz, so that $N=1$. Another essential ingredient is the \emph{entropy tensor}
\begin{equation}
P^{abcd}=\frac{\partial \mathcal{L}}{\partial R_{abcd}}\,.
\label{eq:defentrot}
\end{equation}
It determines the gravitational equations of motion $\mathcal{E}_{ab}$ of the theory under study, which take the following form \cite{Padmanabhan:2011ex,Padmanabhan:2013xyr}:
\begin{equation}
\mathcal{E}_{ab}=P_{acde} R_{b}{}^{cde}-\frac{1}{2}g_{ab} \mathcal{L}+2 \nabla^c \nabla^d P_{acbd}=0\,.
\label{eq:eomgen}
\end{equation}

We have identified up to six possibly different definitions of QTGs which have appeared in the literature over the years. Let us emphasize  that these definitions need not be equivalent among each other (and they will not be in general, as we shall see below). We list them below, indicating the references from which each definition was extracted. 

\begin{notion}[First notion of QTGs \cite{Bueno:2019ltp,Moreno:2023rfl,Moreno:2023arp}]
\label{def:1}
A theory $\mathcal{L}(g^{ab}, R_{cdef})$ is a QTG if
\begin{equation}
\frac{\delta L_{f}}{\delta f}=0 \, \quad  \text{and} \, \quad \left. \frac{\delta L_{N,f}}{\delta N} \right \vert_{N=1} = \frac{\mathrm{d}}{\mathrm{d}r} \left ( \mathcal{F}(r,f) \right) \,,
\end{equation}
where the $\delta$ operation stands for functional differentiation and where $\mathcal{F}$ is a certain function of the indicated variables.
\end{notion}
\begin{notion}[Second notion of QTGs \cite{Bueno:2019ycr,Bueno:2022res}]
\label{def:2}
A theory $\mathcal{L}(g^{ab}, R_{cdef})$ is a QTG if the entropy tensor is divergenceless on \eqref{eq:sssn}:
\begin{equation}
\nabla_a P^{abcd} \vert_{f}=0\,.
\label{eq:conddef2}
\end{equation}
\end{notion}
\begin{notion}[Third notion of QTGs \cite{Myers:2010ru,Dehghani:2011vu}]
\label{def:3}
A theory $\mathcal{L}(g^{ab}, R_{cdef})$ is a QTG if
\begin{equation}\label{defi3}
L_{N,f}=N(r) \frac{\diff}{\diff r} \left ( \mathcal{F}(r,f) \right)+ \frac{\mathrm{d} \mathcal{G}}{\mathrm{d} r}\,,
\end{equation} 
for certain functions\footnote{We do not indicate the variable on which $\mathcal{G}$ depends, since it appears in $L_{N,f}$ as a total derivative and will not contribute to the equations of motion.} $\mathcal{F}$ and $\mathcal{G}$. 
\end{notion}

\begin{notion}[Fourth notion of QTGs \cite{Bueno:2024dgm,Bueno:2024zsx,Bueno:2024eig}]
\label{def:4}
A theory $\mathcal{L}(g^{ab}, R_{cdef})$ is a QTG if the entropy tensor is divergenceless on \eqref{eq:sssintro}:
\begin{equation}
\nabla_a P^{abcd} \vert_{N,f}=0\,.
\end{equation}
\end{notion}
\begin{notion}[Fifth notion of QTGs \cite{Oliva:2010eb,Dehghani:2011vu,Cisterna:2017umf}]
\label{def:5}
A theory $\mathcal{L}(g^{ab}, R_{cdef})$ is a QTG if its equations of motion on SSS backgrounds are second order in derivatives.
\end{notion}

\begin{notion}[Sixth notion of QTGs \cite{Oliva:2010eb,Oliva:2010zd,Oliva:2011xu}]
A theory $\mathcal{L}(g^{ab}, R_{cdef})$ is a QTG if its traced field equations are of second order in derivatives. 
\label{def:6}
\end{notion}

Observe that in some cases different notions were extracted from the same reference, as the authors proved that those notions were equivalent for the particular instances of theories they were considering (for instance, reference \cite{Oliva:2010eb} and Notions \ref{def:5}, \ref{def:6}).

Finally, a significant focus of our work will be the connection of these notions with the fulfillment of a Birkhoff theorem. There exist various formulations of Birkhoff's theorem which are equivalent for vacuum GR \cite{Schmidt:2012wj}. For our purposes, it will be convenient to introduce Birkhoff's theorem as follows.

\begin{defi}
Let $\mathcal{L}(g^{ab},R_{cdef})$ be a theory of gravity. It is said to satisfy a Birkhoff theorem if all spherically symmetric solutions are static\footnote{That is, they admit an asymptotically timelike Killing vector field which is hypersurface-orthogonal.} and uniquely characterized by a single continuous integration constant, together with a possible additional discrete parameter.
\label{def:birk}
\end{defi}

In vacuum GR all spherically symmetric (SS) solutions are isometric to the Schwarzschild-Tangherlini solution and uniquely characterized by their ADM mass, so that Birkhoff's theorem trivially follows. In Lovelock theories and polynomial QTGs fulfilling any of the Notions \ref{def:3}, \ref{def:4}  or \ref{def:5}, all SS solutions are determined by an algebraic equation for $f(t,r)=f(r)$ depending on the ADM mass \cite{Bueno:2019ltp,Bueno:2019ycr,Bueno:2022res,Moreno:2023rfl}, so that SS solutions will also be labeled by an additional discrete parameter. To see this more explicitly, consider the particular case of Einstein-Gauss-Bonnet theory. For a given ADM mass, there will be two different SS solutions, meaning that there is a discrete parameter $q \in \mathbb{Z}_2$ characterizing the solutions.\footnote{Of course, only one solution is connected to the Schwarzschild-Tangherlini solution as the Gauss-Bonnet coupling tends to zero.}

However, there could be theories of gravity for which SSS solutions contain additional free continuous integration parameters beyond the ADM mass. We believe this goes against the spirit of the original GR Birkhoff's theorem, as even SSS solutions would be specified by further continuous parameters beyond the ADM mass --- giving rise to proper hairy black hole solutions. As a consequence, Definition \ref{def:birk} appears to be the most natural notion for the concept of a higher-curvature theory satisfying a Birkhoff theorem.

The characterization of the Birkhoff theorem given in Definition~\ref{def:birk} agrees with the standard usage in the literature. One notable exception is the case of pure Weyl-squared gravity, which has been proven to satisfy a Birkhoff-\textit{like} theorem~\cite{Riegert:1984zz}. In that case the equations of motion enforce that the general spherical solution is static \textit{up to a conformal factor}. Hence, the argument relies crucially on the properties of the pure Weyl-squared action and the result would fail if the Weyl-squared term was coupled to an Einstein-Hilbert term.

\subsection{Main results}

In this subsection we present our findings regarding the interrelations between the notions of QTG defined above. We relegate the proofs of the theorems to future sections of the manuscript, to facilitate the reading of the document and the exploitation of results.

The first result concerns higher-curvature gravities with second-order equations on the single-function ansatz \req{eq:sssn}. %
\begin{theorem}
Notions \ref{def:1} and \ref{def:2} are equivalent. 
\label{thm:1}
\end{theorem}
We then concentrate on theories of gravity which possess second-order equations for general spherically symmetric backgrounds \eqref{eq:sssintro}. We find that:
\begin{theorem}
Notions \ref{def:3}, \ref{def:4} and \ref{def:5} are equivalent. Furthermore, these imply that the equations on general spherical backgrounds are second-order.
\label{thm:2}
\end{theorem}
Interestingly enough, it is sufficient to demand second-order equations on general SSS configurations to have second-order equations \emph{also} for non-static SS solutions. On the other hand, notions \ref{def:3}, \ref{def:4} and \ref{def:5} trivially imply Notions \ref{def:1} and \ref{def:2}. However, the converse is not true: this will be made explicit in Remark \ref{rm:counterex}.  

Third, we focus on general SS backgrounds, removing the staticity condition. Specifically, we will examine which higher-curvature theories satisfy a Birkhoff theorem, as introduced in Definition \ref{def:birk}. We prove the following result:
\begin{theorem}
Let $\mathcal{L}(g^{ab},R_{cdef})$ be a certain theory of gravity built from analytic functions of polynomial curvature invariants.
\begin{enumerate}
    \item If it satisfies Birkhoff's theorem, then the theory fulfills Notions \ref{def:3}, \ref{def:4} and \ref{def:5}.
    \item If all static and spherically symmetric solutions are uniquely characterized by a single continuous integration constant and a possible additional discrete parameter, then the theory satisfies a Birkhoff theorem. 
\end{enumerate}
\label{thmBirk}
\end{theorem}

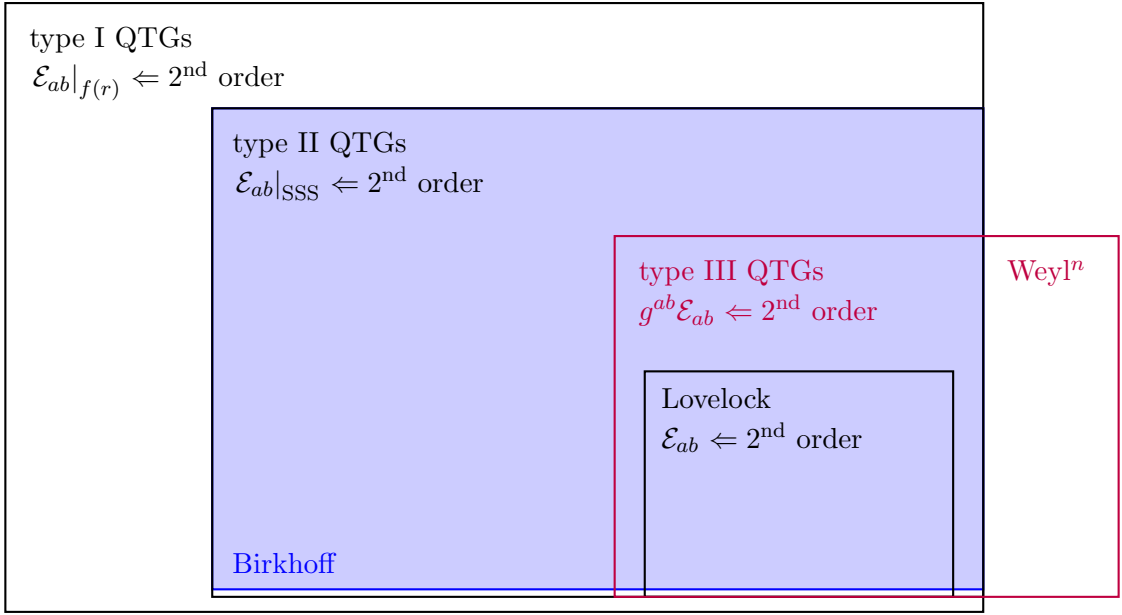
\begin{figure}[t!]
\centering
\begin{tikzpicture}

\node[draw, thick,minimum width=13cm, minimum height=8.1cm, anchor=north west] (R1) at (-1,0) {};
\node[anchor=north west, align=left] at (-0.8,-0.2) {type I QTGs \\ $\left.\mathcal{E}_{ab}\right|_{f(r)}$  $\Leftarrow 2^{\rm nd}$ order};


\node[draw, draw=blue, fill=blue!20,thick, minimum width=10.25cm, minimum height=6.4cm, anchor=north west] (R2) at (1.75,-1.4) {};
\node[anchor=north west, align=left] at (1.9,-7.2) { {\color{blue} Birkhoff}};

\node[draw, draw=black,thick, minimum width=10.25cm, minimum height=6.5cm, anchor=north west] (R2) at (1.75,-1.4) {};
\node[anchor=north west, align=left] at (1.9,-1.6) { type II QTGs \\$\left.\mathcal{E}_{ab}\right|_{\rm SSS}$ $\Leftarrow 2^{\rm nd}$ order};


\node[draw, thick ,minimum width=4.1cm, minimum height=3cm, anchor=north west] (R4) at (7.5,-4.9) {};
\node[anchor=north west, align=left] at (7.6,-5) {Lovelock\\  $\mathcal{E}_{ab}$ $\Leftarrow 2^{\rm nd}$ order  };

\node[draw, thick, draw=purple,minimum width=6.7cm, minimum height=4.8cm, anchor=north west] (R3) at (7.1,-3.1) {};
\node[anchor=north west, align=left] at (7.3,-3.3) { \color{purple} type III QTGs \qquad \qquad \qquad  Weyl$^n$ \\ \color{purple} $g^{ab}\mathcal{E}_{ab}$ $\Leftarrow 2^{\rm nd}$ order};


\end{tikzpicture}
\caption{Schematic representation of the different types of QTGs as classified in this paper. The broadest class, type I QTGs (interior of largest black rectangle), corresponds to theories for which the equations of motion evaluated on the SSS single-function ansatz are second order. It contains the set of type II QTGs (interior of second largest black rectangle), which  are theories for which the equations of motion evaluated on a general SSS ansatz are second order. The type II set  turns out to be almost identical to the set of theories which satisfy a Birkhoff theorem (region shaded in blue). The difference between both is a zero-measure set of type II QTGs which do not satisfy a Birkhoff theorem (represented by a horizontal thin rectangle bounded by black and blue lines). The set of type III QTGs (interior of the purple rectangle) involves theories with second-order traced equations. Observe that it is not completely contained within the type II or type I classes, as the addition of pure Weyl invariants, which possess second-order traced field equations, may take a theory out of the type II and type I sets. The type III and type II sets also contain Lovelock theories, for which the equations of motion are second order in general backgrounds. }
\label{fig:qtg-nested}
\end{figure}

As a consequence, if $\mathcal{L}$ is any higher-curvature gravity built from arbitrary polynomial contractions of Riemann curvature tensors and satisfies a Birkhoff theorem, it must satisfy Notions \ref{def:3}, \ref{def:4} and \ref{def:5}. There is no other option. The converse statement is \emph{almost} true: there exists a zero-measure set of QTGs of this kind that do not satisfy a Birkhoff theorem. This will be clarified in Remark \ref{rm:Birkpart}.

Regarding the relationship between Notion \ref{def:6} and the others, we find the following:
\begin{theorem}
Let $\mathcal{L}(g^{ab},R_{cdef})$ satisfy Notion \ref{def:6}. Then, there exists an analytic function of the Weyl invariants $W_2\equiv W_{ab}^{cd}W^{ab}_{cd} $ and $W_3\equiv W_{ab}^{cd} W_{cd}^{ef} W_{ef}^{ab}$, $\mathcal{L}'=\mathcal{L}'(W_2,W_3)$, such that $\mathcal{L}(g^{ab},R_{cdef})- \mathcal{L}'(W_2,W_3)$ satisfies Notions \ref{def:3}, \ref{def:4} and \ref{def:5}.
\label{thm4}
\end{theorem}

We refer the reader to Remark \ref{rem:cex2ndtrace}, where we show an example of a theory satisfying Notions \ref{def:3}, \ref{def:4} and \ref{def:5} but without second-order traced-field equations, even if one allows for the addition of pure Weyl invariants. 

In sum, the various notions of QTGs identified above reduce to three different classes of higher-curvature theories:
\begingroup
\renewcommand\labelitemi{\ding{117}}
\begin{itemize}
\item \textbf{QTGs of type I}. These satisfy Notions \ref{def:1} and \ref{def:2}.  
\item \textbf{QTGs of type II}. These fulfill Notions \ref{def:3}, \ref{def:4} and \ref{def:5}.  
\item \textbf{QTGs of type III}. These satisfy Notion \ref{def:6}.
\end{itemize}
\endgroup

We will use the nomenclature interchangeably: \emph{type X QTGs} or \emph{QTGs of type X}, where $X=\left\lbrace\right.$I,\,II,\,III$\left. \right\rbrace$. From our previous discussions, it is possible to add/remove pure  Weyl invariants to QTGs of type III to convert them into QTGs of type II. Also, QTGs of type II naturally belong to the type I class. Nevertheless, theories of type I need not be of type II, and theories of this latter type need not correspond to a QTG of type III, even with the addition or subtraction of pure Weyl invariants. Note that Einstein gravity and all non-trivial Lovelock gravities in a given dimension are QTGs of type III that also belong to the type II class.

Observe that QTGs of type III may actually provide higher-order equations on spherical backgrounds --- just take any pure Weyl invariant. Similarly, QTGs of type I may yield in general higher-derivative equations on general (S)SS configurations. Unlike them, QTGs of type II do provide second-order equations for these space-times, being actually the unique theories of gravity fulfilling a Birkhoff theorem --- see Theorem \ref{thmBirk}. Furthermore, QTGs of type II are known to exist at all curvature orders \cite{Bueno:2019ycr,Bueno:2024dgm,Bueno:2024zsx} and can be mapped into any effective theory of gravity by perturbative field redefinitions of the metric \cite{Bueno:2019ltp,Bueno:2024dgm}. Therefore, it seems that QTGs of type II represent, by far, the most natural notion of quasi-topological gravity.

\section{Theories with $2^{\rm nd}$-order equations on single-function SSS backgrounds}\label{sec2}

Let us consider the following single-function SSS configurations:
\begin{equation}
\mathrm{d}s_{f}^2=-f(r) \mathrm{d} t^2+\frac{\mathrm{d} r^2}{f(r)}+r^2 \mathrm{d} \Omega_{D-2}^2\,.
\label{eq:sssn1}
\end{equation}
The Riemann curvature tensor of \eqref{eq:sssn1} takes the following form:
\begin{equation}
\left. R_{ab}{}^{cd} \right \vert_f=-2f'' \tau_{[a}^{[c}  \rho_{b]}^{d]} -\frac{2 f'}{r} \left (\tau_{[a}^{[c}+\rho_{[a}^{[c} \right ) \sigma_{b]}^{d]}+2 \frac{(1-f)}{r^2} \sigma_{[a}^{[c}  \sigma_{b]}^{d]}\,,
\label{eq:Rformf}
\end{equation}
where $\tau_{a}^b=\delta^t_a\delta^b_t$, $\rho_a{}^b=\delta^r_a\delta^b_r$ and $\sigma_a^b=\delta_a^b-\tau_a^b-\rho_a^b$ stand for the projectors into the temporal, radial and the angular components, respectively. These projectors satisfy:
\begin{equation}
\label{eq:proy}
\tau_{a}^b \tau_b^c=\tau_a^c \,, \quad  \rho_{a}^b \rho_b^c=\rho_a^c\,, \quad  \sigma_a^b \sigma_b^c=\sigma_a^c\,, \quad \tau_a^b \rho_b^c= \sigma_a^b \tau_b^c=\sigma_a^b \rho_b^c=0\,, \quad\tau_a^a=\rho_a^a=\frac{\sigma_a^a}{D-2}=1\,.
\end{equation}
Such a structure for the Riemann curvature tensor enforces the following expression for the entropy tensor \cite{Bueno:2019ycr}:
\begin{equation}
\label{eq:Pf}
\left. P_{ab}{}^{cd} \right \vert_f=P^{(1)}_f (r) \tau_{[a}^{[c}  \rho_{b]}^{d]} +P^{(2)}_f (r) \left (\tau_{[a}^{[c}+\rho_{[a}^{[c} \right )  \sigma_{b]}^{d]}+P^{(3)}_f (r) \sigma_{[a}^{[c}  \sigma_{b]}^{d]}\,,
\end{equation}
where
\begin{equation}
P^{(1)}_f=-2\frac{\partial \mathcal{L}_f}{\partial f''}\,, \quad P^{(2)}_f=-\frac{r}{(D-2)}\frac{\partial \mathcal{L}_f}{\partial f'}\,, \quad P^{(3)}_f=-\frac{r^2}{(D-2)(D-3)}\frac{\partial \mathcal{L}_f}{\partial f}\,.
\label{eq:Pformf}
\end{equation}
Direct computations show that:
\begin{align}
\left.\nabla_e P_{ab}{}^{ed} \right \vert_f&=A_f\,  \delta_{[a}{}^r\tau_{b]}{}^d + B_f  \, \delta_{[a}{}^r \sigma_{b]}{}^d\,, \\ \nonumber
\left. \nabla^a \nabla_e P_{ab}{}^{ed} \right \vert_f&=\left[\frac{f A_f'}{2}+\frac{(r f'+2(D-2)f)A_f}{4r}  \right] \tau_b{}^d+ \left[ \frac{f' A_f}{4}+\frac{(D-2)f B_f}{2r}\right] \rho_b{}^d \\ \label{eq:ccPf} & +\left[\frac{f B_f'}{2}+\frac{(r f'+(D-3)f)B_f}{2r}  \right] \sigma_b{}^d\,,
\end{align}
where we have defined:
\begin{equation}
A_f=\frac{1}{2}\frac{\mathrm{d} P^{(1)}_f}{\mathrm{d}r}+\frac{(D-2)\left (P^{(1)}_f-P^{(2)}_f \right )}{2r}\,, \quad B_f=\frac{1}{2}\frac{\mathrm{d} P^{(2)}_f}{\mathrm{d}r}+\frac{(D-3)\left  (P^{(2)}_f-2P^{(3)}_f \right)}{2r} \,.
\end{equation}
For future purposes, it is convenient to define 
\begin{equation}
\psi\equiv \frac{1-f}{r^2}\,.
\end{equation}
\begin{proposition}
\label{prop1}
Let $\mathcal{L}(g^{ab},R_{cdef})$ be a higher-curvature gravity built from polynomial combinations of curvature invariants. Assume its equations of motion on \eqref{eq:sssn1} contain at most second derivatives of $f(r)$. Then, the evaluation of $\mathcal{L}(g^{ab},R_{cdef})$ on the background \eqref{eq:sssn1}, denoted as $\mathcal{L}_f$, takes the form:
\begin{equation}
\label{eq:laghf}
\mathcal{L}_f=h_0 \left (\psi \right) f''+\frac{\partial h_0 (\psi)}{\partial f}(f')^2+h_1 \left (\psi \right)\frac{f'}{r}+h_2 \left (\psi \right)\,, 
\end{equation}
for certain single-variable functions $h_0,h_1$ and $h_2$. 
\end{proposition}
\begin{proof}
Assume that the equations of motion of $\mathcal{L}(g^{ab},R_{cdef})$ on single-function SSS backgrounds \eqref{eq:sssn1} contain at most second-order derivatives of $f(r)$. This implies that the term $\nabla^a \nabla_e P_{ab}{}^{ed}$ (the source of higher derivatives in the equations of motion) does not provide any higher-derivatives of $f(r)$ upon evaluation on \eqref{eq:sssn1}. 

We already presented $\left. \nabla^a \nabla_e P_{ab}{}^{ed} \right \vert_f$ in \eqref{eq:ccPf}, so we just need to study the sufficient and necessary conditions for fourth and third derivatives of $f(r)$ to be absent. Let us first focus on fourth derivatives. Indeed, these may only appear from those terms which correspond to second derivatives of $P_f^{(1)}$ and $P_f^{(2)}$ (following the notation introduced in \eqref{eq:Pf}). As a matter of fact, by direct computation one may check that:
\begin{equation}
\frac{\partial \left. \nabla^a \nabla_e P_{ab}{}^{ed} \right \vert_f}{\partial f''''}=-\frac{f}{2} \frac{\partial^2 \mathcal{L}_f}{\partial (f'')^2}\tau_b{}^d-\frac{r f}{4(D-2)}  \frac{\partial^2 \mathcal{L}_f}{\partial f' \partial f''}\sigma_b{}^d\,.
\end{equation}
Therefore, no fourth derivatives in the equations of motion on single-function SSS backgrounds requires:
\begin{equation}
\label{eq:no4df}
\frac{\partial^2 \mathcal{L}_f}{\partial (f'')^2}=\frac{\partial^2 \mathcal{L}_f}{\partial f' \partial f''}=0\,.
\end{equation}
This condition forces $\mathcal{L}_f$ to adopt the form:
\begin{equation}
\label{eq:lfprov}
\mathcal{L}_f=h_0(r,f)f''+\mathcal{H} (r,f,f')\,,
\end{equation}
where $h_0$ and $\mathcal{H}$ are certain functions depending on the indicated variables. Let us now obtain the conditions for the absence of third derivatives of $f(r)$. To this aim, we compute:
\begin{equation}
\frac{\partial \left. \nabla^a \nabla_e P_{ab}{}^{ed} \right \vert_f}{\partial f'''}=-\frac{r f}{4(D-2)} \left[ \frac{\partial^2 \mathcal{H}}{\partial (f')^2}-2\frac{\partial h_0}{\partial f}\right] \sigma_b{}^d\,.
\end{equation}
The latter vanishes if and only if:
\begin{equation}
\mathcal{H}(r,f,f')=\frac{\partial h_0 (r,f)}{\partial f} (f')^2+h_1(r,f) f'+h_2(r,f)\,,
\end{equation}
where $h_1$ and $h_2$ are undetermined functions of the indicated variables.

However, let us recall that $\mathcal{L}_f$ is obtained from the combination of curvature invariants evaluated on the single-function SSS ansatz \eqref{eq:sssn1}. Observe the form of the Riemann curvature tensor on \eqref{eq:sssn1} that was presented in \eqref{eq:Rformf}. Up to numerical prefactors, any component of this Riemann tensor takes one of the following forms:
\begin{equation}
\left\lbrace f''\,, \quad \frac{f'}{r}\,, \quad \psi=\frac{1-f}{r^2}\right\rbrace\,.
\end{equation}
As a consequence, $h_0(r,f)$ may only depend on $(r,f)$ through $\psi$. By similar arguments, $h_1$ and $h_2$ must be single-variable functions of $\psi$. Consequently:
\begin{equation}
\mathcal{L}_f=h_0 \left (\psi\right) f''+\frac{\partial h_0 (\psi)}{\partial f}(f')^2+h_1 \left (\psi \right)\frac{f'}{r}+h_2 \left (\psi \right)\,, 
\end{equation}
and we conclude.
\end{proof}

Our next goal will be to restrict the expression of the functions $h_0, h_1$ and $h_2$ in \eqref{eq:laghf}. To this aim, it will be convenient to have at hand the expressions for the Weyl curvature tensor $W_{abcd}$, the traceless Ricci tensor $Z_{ab}=R_{ab}-1/D g_{ab}R$ and Ricci scalar $R$ on \eqref{eq:sssn1}. From the expression of the Weyl curvature tensor, we get \cite{Moreno:2023rfl}:
\begin{align}
\label{eq:weylsimf}
\left.W_{ab}{}^{cd}\right\vert_f&=\Omega_f \left[ (D-2)(D-3) \tau_{[a}^{[c} \rho_{b]}^{d]}-(D-3)\left (\tau_{[a}^{[c}+\rho_{[a}^{[c} \right ) \sigma_{b]}^{d]}+ \sigma_{[a}^{[c} \sigma_{b]}^{d]} \right]\, , \\ 
\left.Z_a^b\right\vert_f&=\Theta_f\left[ -\frac{D-2}{2} \left (\tau_a^b+\rho_a^b \right)+ \sigma_a^b \right ]\, , \quad \left. R \right \vert_f= \Rho_f\,.
\label{eq:ricimf}
\end{align}
where
\begin{align}
\label{eq:omsimf}
\Omega_f=\frac{4-4f+4 r f'-2r^2 f''}{(D-1)(D-2) r^2}&\, , \quad \Theta_f=\frac{2(D-3)(1-f)+(D-4) r f'
+r^2 f''}{D r^2}\,, \\
\Rho_f=&\frac{(D-2)((D-3)(1-f)-2r f')-r^2 f''}{r^2}\,.
\label{eq:scalsimf}
\end{align}

\begin{proposition}
\label{prop2}
Let $\mathcal{L}(g^{ab},R_{cdef})$ be a higher-curvature gravity built from polynomial combinations of curvature invariants. Assume its equations of motion on \eqref{eq:sssn1} are of second order in derivatives. The functions $h_0$, $h_1$ and $h_2$ in \eqref{eq:laghf} can be entirely parametrized as follows:
\begin{equation}
\label{eq:formhf}
h_0(\psi)=\frac{\psi^{(D-2)/2}}{2} \int  \frac{h'(\psi)}{\psi^{D/2}}\mathrm{d}\psi\,, \quad h_1(\psi)=-2 h'(\psi)\,, \quad h_2(\psi)=(D-1)h(\psi)-2\psi h'(\psi)\,,
\end{equation}
where $h$ is an arbitrary analytic function of $\psi$.
\end{proposition}
\begin{proof}
Using \eqref{eq:omsimf} and \eqref{eq:scalsimf}, we may express $f(r)$ and its derivatives in terms of the functions $\Omega_f$, $\Theta_f$ and $\Rho_f$ determining any curvature invariant on the single-function SSS ansatz \eqref{eq:sssn1}:
\begin{align}
\frac{1-f}{r^2}&=\frac{\Rho_f}{D(D-1)}+\frac{\Omega_f}{2}+\frac{2 \Theta_f}{D-2}\,, \\
\frac{f'}{r}&=\frac{2\Rho_f}{D(1-D)}+\frac{(D-3)\Omega_f}{2}+\frac{(D-4) \Theta_f}{D-2}\,, \\
f''&=\frac{2\Rho_f}{D(1-D)}-\frac{(D-2)(D-3) \Omega_f}{2} +2\Theta_f\,.
\end{align}
Substituting these results on the form of the Lagrangian \eqref{eq:laghf}, we obtain $\mathcal{L}_f$ in terms of $\Omega_f$, $\Theta_f$ and $\Rho_f$. However, for $\mathcal{L}_f$ to arise from an actual combination of curvature invariants, some necessary conditions must hold. In particular, if one expands $\mathcal{L}_f$ for sufficiently small\footnote{We are assuming that $\mathcal{L}_f$ may be written as a power series in the variables $\Omega_f$, $\Theta_f$ and $\Rho_f$, which is consistent with the fact that the theory is constructed from arbitrary polynomial combinations of curvature invariants.} $\Omega_f$ and $\Theta_f$, no terms of the form $g_1(\Rho_f) \Omega_f$ or  $g_2(\Rho_f) \Theta_f$ (with $g_1$ and $g_2$ arbitrary functions) may arise in $\mathcal{L}_f$, as these pieces cannot come from polynomial contractions of curvature tensors. Therefore, examining the expansion of $\mathcal{L}_f$ for small $\Omega_f$ and $\Theta_f$:
\begin{align}
\label{eq:lexpanf}
\mathcal{L}_f&=\mathcal{L}_f^{(0)}(\Rho_f)+\mathcal{E}_\Omega\frac{\Omega}{2}+ \frac{\mathcal{E}_\Theta}{D-2}\Theta+ \frac{\mathcal{E}_{\Omega \Theta}}{2(D-2)}\, \Omega \Theta+\mathcal{O}(\Omega^2, \Theta^2)\,, \\
\mathcal{E}_\Omega(x)&=(D-3) \left ( h_1(x)-(D-2) h_0(x)\right) +h_2'(x)+2x \left ( (2D-7) h_0'(x)-h_1'(x)-2 x h_0''(x) \right)\,, \\
\mathcal{E}_\Theta (x) &=4 x \left((D-5) h_0'(x)-2 x h_0''(x)-h_1'(x)\right)+2 (D-2) h_0(x)+(D-4) h_1(x)+2 h_2'(x)\,, \\ \nonumber
\mathcal{E}_{\Omega \Theta} (x)&=(3D-10)h_1'(x)+2h_2''(x)-2(D-4)(2D-5)h_0'(x)\\& -4x(2x h_0'''(x)+h_1''-(3D-11)h_0'')\,,
\end{align}
where $x=\frac{\Rho_f}{D(D-1)}$. Now:
\begin{equation}
\label{eq:h1f}
2 \mathcal{E}_\Omega(x)- \mathcal{E}_\Theta(x)=(D-2)(-2(D-2)h_0(x)+h_1(x)+4x h_0'(x))\,.
\end{equation}
Setting the previous expression to zero, we may solve for $h_1(x)$ and find:
\begin{equation}
h_1(x)=2(D-2)h_0(x)-4x h_0'(x)\,.
\end{equation}
If we substitute this constraint into $\mathcal{E}_\Omega$ and set the resulting expression to zero:
\begin{equation}
\label{eq:h2f}
(D-2)(D-3)h_0(x)-2(2D-7) x h_0'(x)+h_2'(x)+4 x^2 h_0''(x)=0\,.
\end{equation}
Interestingly enough, conditions \eqref{eq:h1f} and \eqref{eq:h2f} ensure that $\mathcal{E}_\Omega=\mathcal{E}_\Theta=\mathcal{E}_{\Omega\Theta}=0$, so all potentially inadmissible terms in  \eqref{eq:lexpanf} vanish identically. Making the following redefinitions:
\begin{equation}
h_0(x)=\frac{x^{(D-2)/2}}{2} \int  \frac{h'(x)}{x^{D/2}}\mathrm{d}x\,, \quad h_1(x)=2 g(x)\,, \quad h_2(x)=(D-1) w(x)-2x w'(x)
\end{equation}
in terms of some new functions $(h,g,w)$, conditions \eqref{eq:h1f} and \eqref{eq:h2f} are rephrased as:
\begin{equation}
g(x)=-h'(x)\,, \quad  w'(x)-h'(x)+\frac{2x}{D-3}(h''(x)-w''(x))=0\,.
\end{equation}
The latter expression can be straightforwardly integrated to yield:
\begin{equation}
w(x)=c_0 x^{(D-1)/2}-\frac{2\lambda}{D-2} +h(x)\,,
\end{equation}
where $c_0$ and $\lambda$ are integration constants. Observe that $c_0$ is irrelevant, since it does not affect the expression for $h_2$ in \eqref{eq:laghf}, so we may safely set it to zero, while $\lambda$ represents a cosmological constant that we introduce in the definition of $h(\psi)$. We substitute these results in \eqref{eq:laghf} and we conclude. Note that the analyticity of $h(\psi)$ follows from the theory being constructed from polynomial combinations of curvature invariants.
\end{proof}

We are now able to prove our main result regarding single-function SSS solutions:

\begin{proposition}
Let $\mathcal{L}(g^{ab}, R_{cdef})$ be a higher-curvature theory built from polynomial contractions of curvature tensors. The following are equivalent:
\begin{enumerate}
\item The entropy tensor evaluated on \eqref{eq:sssn1} is divergenceless, $\nabla_a P^{abcd} \vert_f=0$.
\item $\mathcal{L}(g^{ab},R_{cdef})$ possesses second-order equations of motion for the single-function static and spherically symmetric ansatz \eqref{eq:sssn1}.
\item The Lagrangian $\mathcal{L}_f$ evaluated on \eqref{eq:sssn1} is given by:
\begin{equation}
\label{eq:totderf}
\mathcal{L}_f=\frac{1}{r^{D-2}} \frac{\mathrm{d}}{\mathrm{d} r} \left ( \frac{(1-f)^{(D-2)/2}f'}{2} \int^\psi h'(\kappa) \kappa^{-D/2} \mathrm{d}\kappa+r^{D-1} h\left (\frac{1-f}{r^2} \right) \right )\,,
\end{equation}
where $r^2 \psi=1-f$, $h$ is an analytic function and the integration limits of $\int^\psi h'(\kappa) \kappa^{-D/2} \mathrm{d}\kappa$ for odd $D$ are to be taken as those ensuring the absence of fractional powers of $(1-f)$. 
\end{enumerate}
\label{prop:sssn1res}
\end{proposition}

\begin{proof}

``\emph{1} $\rightarrow$ \emph{2}''. Trivial.

``\emph{2} $\rightarrow$ \emph{3}''. By direct computation, using Proposition \ref{prop2} together with \eqref{eq:formhf}.

``\emph{3} $\rightarrow$ \emph{1}''. Starting from \eqref{eq:totderf}, one ends up with the expression for $\mathcal{L}_f$ presented in \eqref{eq:laghf}, where the functions $h_0$, $h_1$ and $h_2$ are given in \eqref{eq:formhf}. Such a Lagrangian gives rise to second-order equations of motion for $f(r)$ and we conclude.
\end{proof}

Note that no integration limits for the integral in \eqref{eq:totderf} were set for even $D$. The reason is that the terms
\begin{equation}
\mathcal{L}_f^{\mathrm{cL}}=\frac{1}{r^{D-2}} \frac{\mathrm{d}}{\mathrm{d} r} \left ((1-f)^{(D-2)/2}f' \right)\,,
\end{equation}
do not affect at all the subsequent equations of motion. As a matter of fact, we have checked that $\mathcal{L}_f^{\mathrm{cL}}$ is nothing but the evaluation of Lovelock densities of order $n$ in \eqref{eq:sssn1} in the critical dimension $D=2n$, up to a constant proportionality factor.

\begin{proposition}
\label{prop:eomsssn1}
Let $\mathcal{L}(g^{ab}, R_{cdef})$ be a higher-curvature theory constructed from polynomial contractions of curvature tensors. If the equations of motion on \eqref{eq:sssn1} feature no derivatives of $f$ of degree higher than two, then the unique single-function solutions \eqref{eq:sssn1} are determined by a first-order differential equation for $f(r)$ which can be exactly integrated to yield:
\begin{equation}
\label{eq:eomffinal}
h\left (\frac{1-f}{r^2} \right)=\frac{2\mathsf{M}}{r^{D-1}}\,,
\end{equation}
where $\mathsf{M}$ is an integration constant related to the ADM mass of the solution and $h$ is a theory-dependent function which may be computed as
\begin{equation}
(D-1)h\left (\frac{1-f}{r^2} \right)=\left. \mathcal{L}_f \right \vert_{f' \rightarrow 0,\, f''\rightarrow 0} - \left. r \psi \frac{\partial \mathcal{L}_f}{\partial f'}  \right \vert_{f' \rightarrow 0}\,,
\label{eq:obtenerh}
\end{equation}
where $f$, $f'$ and $f''$ are here conceived as if they were independent variables ---  \ie sending $f' \rightarrow 0$ or $f'' \rightarrow 0$ does not imply anything on $f$ in the matter of this computation.
\end{proposition}
\begin{proof}
From \eqref{eq:Pformf} and \eqref{eq:Rformf} we have that:
\begin{align}
\nonumber
\left. P_{ae}{}^{gh} R_{gh}{}^{be} \right \vert_f&= -\left (\frac{P^{(1)}_f f''}{4}+\frac{(D-2)P^{(2)}_f f'}{4r} \right) \left (\tau_{a}{}^b+\rho_{a}{}^b \right)\\&-\left (\frac{P^{(2)}_f f'}{2r}-\frac{(D-3) P^{(3)}_f (1-f)}{r^2} \right) \sigma_a{}^b\,.
\end{align}
Using now Proposition \ref{prop:sssn1res}, the gravitational equations of motion along the $tt$ and $rr$ components are the same and given by:
\begin{equation}
\frac{P^{(1)}_f f''}{4}+\frac{(D-2)P^{(2)}_f f'}{4r}+\frac{1}{2} \mathcal{L}_f=0\,.
\end{equation}
Using Proposition \ref{prop2} (in particular the  expressions \eqref{eq:formhf}) one finds the equation:
\begin{equation}
(D-1) r h(\psi)-(2r \psi+f') h'(\psi)=0\,.
\end{equation}
This can be directly expressed as:
\begin{equation}
\frac{1}{r^{D-3}}\frac{\mathrm{d}}{\mathrm{d}r} \left ( r^{D-1} h(\psi) \right)=0 \implies \quad h(\psi)=\frac{2\mathsf{M}}{r^{D-1}}\,,
\end{equation}
where $\mathsf{M}$ is an integration constant proportional to the mass of the subsequent solution. The angular components of the equations of motion $\mathcal{E}_{ab}$ are automatically solved by virtue of the contracted Bianchi identity. Indeed, the structure of the curvature tensor on \eqref{eq:sssn1} forces the equations of motion of a theory satisfying any of the items in Proposition \ref{prop2} to be:
\begin{equation}
\mathcal{E}_{ab}=\mathcal{E}_{\rm rad} (\tau_{ab}+\rho_{ab})+\mathcal{E}_{\rm ang} \sigma_{ab}\,.
\end{equation}
Now, the condition $\nabla^a \mathcal{E}_{ab}$ is equivalent to:
\begin{equation}
\frac{\mathrm{d}\mathcal{E}_{\rm rad}}{\mathrm{d}r}+\frac{(D-2) (\mathcal{E}_{\rm rad}-\mathcal{E}_{\rm ang})}{r}=0\,.
\label{eq:bianchisssn1}
\end{equation}
Therefore, $\mathcal{E}_{\rm rad}=0$ implies $\mathcal{E}_{\rm ang}=0$. Finally, \eqref{eq:obtenerh} follows directly from Proposition \ref{prop2} and we conclude.

\end{proof}

Interestingly enough, Proposition \ref{prop:eomsssn1} implies that if the subsequent equations of motion contain no higher derivatives when considered on the single-function SSS ansatz \eqref{eq:sssn1}, then the equations of motion do admit solutions \eqref{eq:sssn1} in which the equation for $f(r)$ can be immediately integrated into an algebraic equation. Observe that this feature is highly non-trivial, as SSS solutions are generically determined by two independent functions.

We have collected all the relevant material (including some additional interesting facts) required to prove Theorem \ref{thm:1}.

{\noindent  \bf Proof of Theorem \ref{thm:1}.} Observe that the relation between the components $tt$ and $rr$ of the gravitational equations of motion $\mathcal{E}_{ab}$ with the functional derivatives of $L_{N,f}=r^{D-2} N \mathcal{L}_{N,f}$ is given by:
\begin{equation}
\left. \frac{\delta L_{f}}{\delta f}\right \vert_f=\frac{r^{D-2}}{f} \left(-\mathcal{E}_{t}^{t}+\mathcal{E}_{r}^{r}\right)\,, \quad \left. \frac{\delta L_{N,f}}{\delta N} \right \vert_f=-2r^{D-2} \mathcal{E}_{t}^{t}\,.
\label{eq:derfuncond}
\end{equation}
This is directly obtained by using the chain rule and the expression for the most general SSS metric \eqref{eq:sssintro}. Now, if we assume Notion \ref{def:1} holds, then $\mathcal{E}_{r}^{r}=\mathcal{E}_{t}^{t}$. By the contracted Bianchi identity  (cf. \eqref{eq:bianchisssn1}), one infers that the angular components are of second derivative order, so that the whole set of equations of motion is second order and by Proposition \ref{prop:sssn1res}, $\nabla_a P^{abcd} \vert_f=0$ is satisfied.

Conversely, if $\nabla_a P^{abcd} \vert_f=0$, Proposition \ref{prop:sssn1res} ensures that $L_f$ is a total derivative, so $\frac{\delta L_{f}}{\delta f}=0$ and $\mathcal{E}_{t}^{t}=\mathcal{E}_{r}^{r}$. By virtue of the contracted Bianchi identity, $\mathcal{E}_{r}^{r}$ must be of first derivative order, admitting a direct integration into an algebraic equation --- see the proof of Proposition \ref{prop:eomsssn1}. Hence we conclude. \qed

\section{Theories with $2^{\rm nd}$-order equations on general SSS backgrounds}\label{sec3}

Let us now study  general SSS configurations:
\begin{equation}
\mathrm{d}s_{N,f}^2=-N(r)^2f(r) \mathrm{d} t^2+\frac{\mathrm{d} r^2}{f(r)}+r^2 \mathrm{d} \Omega_{D-2}^2\,.
\label{eq:sss}
\end{equation}
The Riemann curvature tensor of \eqref{eq:sss} reads as follows:
\begin{equation}
\left. R_{ab}{}^{cd} \right \vert_{N,f}=\mathcal{R}_{N,f}^{(1)}\tau_{[a}^{[c}  \rho_{b]}^{d]}+\mathcal{R}_{N,f}^{(2)} \tau_{[a}^{[c}  \sigma_{b]}^{d]} +\mathcal{R}_{N,f}^{(3)} \rho_{[a}^{[c} \sigma_{b]}^{d]}+\mathcal{R}_{N,f}^{(4)} \sigma_{[a}^{[c}  \sigma_{b]}^{d]}\,,
\label{eq:RformNf}
\end{equation}
where we have defined
\begin{align}
\nonumber
\mathcal{R}_{N,f}^{(1)}&=-\frac{6 f' N'+2N f''+4f N''}{N}\,, \quad \mathcal{R}_{N,f}^{(2)}=-\frac{2N f'+4f N'}{r N}\,,\\ \mathcal{R}_{N,f}^{(3)}&=-\frac{2 f'}{r} \,, \quad \mathcal{R}_{N,f}^{(4)}=2 \frac{(1-f)}{r^2} \,.
\label{eq:RformNfcomp}
\end{align}
The form of $\left. R_{ab}{}^{cd} \right \vert_{N,f}$ induces the following structure for the entropy tensor on \eqref{eq:sss}:
\begin{equation}
\label{eq:PNf}
\left. P_{ab}{}^{cd} \right \vert_{N,f}=P_{N,f}^{(1)} \tau_{[a}^{[c}  \rho_{b]}^{d]} +P_{N,f}^{(2)} \tau_{[a}^{[c}\sigma_{b]}^{d]}+P_{N,f}^{(3)}\rho_{[a}^{[c}  \sigma_{b]}^{d]}+P_{N,f}^{(4)}\sigma_{[a}^{[c}  \sigma_{b]}^{d]}\,,
\end{equation}
where
\begin{equation}
P_{N,f}^{(1)}=4\frac{\partial \mathcal{L}_{N,f}}{\partial \mathcal{R}_{N,f}^{(1)}}\,,  \quad P_{N,f}^{(4)}=\frac{2}{(D-2)(D-3)}\frac{\partial \mathcal{L}_{N,f}}{\partial \mathcal{R}_{N,f}^{(4)}}\,, \quad  P_{N,f}^{(I)}=\frac{4}{D-2}\frac{\partial \mathcal{L}_{N,f}}{\partial \mathcal{R}_{N,f}^{(I)}}\,, \quad I=2,3 \,.
\label{eq:pissss}
\end{equation}
In terms of the derivatives of $f$ and $N$:
\begin{align}
\label{eq:PformNf1}
P_{N,f}^{(1)}&=-2\frac{\partial \mathcal{L}_{N,f}}{\partial f''}\,, \quad P_{N,f}^{(2)}=-\frac{r}{(D-2)f}\left[ N\frac{\partial \mathcal{L}_{N,f}}{\partial N'}-3 f' \frac{\partial \mathcal{L}_{N,f}}{\partial f''} \right]\,, \\ \label{eq:PformNf2} P_{N,f}^{(3)}&=-P_{N,f}^{(2)}-\frac{2r}{(D-2)} \left[\frac{\partial \mathcal{L}_{N,f}}{\partial f'}-\frac{3 N'}{N} \frac{\partial \mathcal{L}_{N,f}}{\partial f''} \right]\,, \\ P_{N,f}^{(4)}&=-\frac{r N'}{(D-3) N} P_{N,f}^{(2)} -\frac{r^2}{(D-2)(D-3)}\left[\frac{\partial \mathcal{L}_{N,f}}{\partial f}-\frac{2 N''}{N} \frac{\partial \mathcal{L}_{N,f}}{\partial f''} \right]\,.
\label{eq:PformNf3}
\end{align}
It is convenient to compute the following:
\begin{align}
\left.\nabla_e P_{ab}{}^{ed} \right \vert_{N,f}&=A_{N,f}\,  \delta_{[a}{}^r\tau_{b]}{}^d  + B_{N,f} \, \delta_{[a}{}^r \sigma_{b]}{}^d \,, \\ 
\nonumber
\left. \nabla^a \nabla_e P_{ab}{}^{ed} \right \vert_{N,f}&= \left[\frac{fA_{N,f}'}{2}+\frac{(rf'+2(D-2)f)A_{N,f}}{4r} \right] \tau_b{}^d+  \frac{(2fN'+f' N)A_{N,f}}{4 N}\rho_b{}^d\\  &  \label{eq:nnPSSS} + \frac{(D-2)f B_{N,f}}{2r}\rho_b{}^d+ \left[ \frac{f B_{N,f}'}{2}+ \frac{(r f'+(D-3)f)B_{N,f}}{2r}+\frac{f N' B_{N,f}}{2N} \right] \sigma_b{}^d\,,
\end{align}
where we have introduced:
\begin{align}
\nonumber
A_{N,f}&=\frac{1}{2}\frac{\mathrm{d} P^{(1)}_{N,f}}{\mathrm{d}r}+\frac{(D-2)\left (P^{(1)}_{N,f}-P^{(2)}_{N,f} \right )}{2r}\,, \\
\label{eq:ABsss} B_{N,f}&=\frac{1}{2}\frac{\mathrm{d} P^{(3)}_{N,f}}{\mathrm{d}r}+\frac{(D-3)\left  (P^{(3)}_{N,f}-2P^{(4)}_{N,f} \right)}{2r}+ \frac{(2fN'+f'N) (P^{(3)}_{N,f}-P^{(2)}_{N,f})}{4fN} \,.
\end{align}
\begin{proposition}
\label{prop:ngen}
Let $\mathcal{L}(g^{ab},R_{cdef})$ be a higher-curvature gravity built from polynomial curvature invariants with second-order equations on a general static and spherically symmetric ansatz \eqref{eq:sss}. Then $\mathcal{L}_{N,f}$ reads as follows:
\begin{align}
\nonumber
\mathcal{L}_{N,f}&=h_0(\psi) f''+\frac{\partial h_0(\psi)}{\partial f}(f')^2+h_1(\psi)\frac{f'}{r}+h_2\left ( \psi \right)   \\ &+\frac{f' N'}{N} \left (3h_0(\psi) +2f \frac{\partial h_0 (\psi)}{\partial f} \right) +\frac{f N'}{r N}h_1(\psi) + \frac{2f}{N}h_0(\psi) N''\,,
\label{eq:lagnftot}
\end{align}
where the functions $h_0$, $h_1$ and $h_2$ are evaluated at $\psi=\dfrac{1-f}{r^2}$ and are given by Proposition \ref{prop2}.
\end{proposition}
\begin{proof}
If $\mathcal{L}(g^{ab},R_{cdef})$ does not give rise to higher derivatives on general SSS backgrounds, then  $\left. \nabla^a \nabla_e P_{ab}{}^{ed} \right \vert_{N,f}$ must not possess any higher derivatives of $f$ and $N$. 

Let us focus first on fourth derivatives. Regard $\mathcal{L}_{N,f}$ as a function of the variables $\mathcal{R}_{N,f}^{(K)}$ with $K=1,2,3,4$ (this is indeed the case by the structure of the curvature tensor \eqref{eq:RformNf}). Since second derivatives of $f$ and $N$ in the Lagrangian $\mathcal{L}_{N,f}$ are entirely encoded in the piece $\mathcal{R}_{N,f}^{(1)}$, it is immediate to see that no fourth derivatives will be present if $\left. \nabla^a \nabla_e P_{ab}{}^{ed} \right \vert_{N,f}$ does not depend on $\dfrac{\mathrm{d}^2 \mathcal{R}_{N,f}^{(1)}}{\mathrm{d}r^2}$. This condition imposes:
\begin{equation}
\frac{f}{4} \frac{\partial P_{N,f}^{(1)}}{\partial \mathcal{R}_{N,f}^{(1)} } \tau_a{}^b+\frac{f}{4} \frac{\partial P_{N,f}^{(3)}}{\partial  \mathcal{R}_{N,f}^{(1)} } \sigma_a{}^b=0\,.
\end{equation}
From here one learns that:
\begin{equation}
\frac{\partial^2 \mathcal{L}_{N,f}}{\partial \left ( \mathcal{R}_{N,f}^{(1)}\right)^2 }=\frac{\partial^2 \mathcal{L}_{N,f}}{\partial  \mathcal{R}_{N,f}^{(1)}  \mathcal{R}_{N,f}^{(3)}}=0\,.
\end{equation}
Furthermore, taking into account that $\mathcal{L}_{N,f}$ must be symmetric under the exchange of the variables $\mathcal{R}_{N,f}^{(2)}$ and $\mathcal{R}_{N,f}^{(3)}$ (this is easily seen\footnote{Indeed, if one was to replace in any curvature invariant all curvature tensors by `\emph{fake}' curvature tensors which have the form \eqref{eq:RformNf} \emph{but} exchanging  $\tau_{a}^b \longleftrightarrow \rho_a^b$, the resulting invariant constructed with these `fake' curvature tensors would be exactly the same.} by the form of the curvature tensor \eqref{eq:RformNf}), we conclude that:
\begin{equation}
\mathcal{L}_{N,f}=\mathcal{H}_0\left (\mathcal{R}_{N,f}^{(4)} \right) \mathcal{R}_{N,f}^{(1)}+\mathcal{H}_1\left ( \mathcal{R}_{N,f}^{(2)},\mathcal{R}_{N,f}^{(3)},\mathcal{R}_{N,f}^{(4)}\right)\,,
\label{eq:lagnfproceso}
\end{equation}
for some arbitrary functions $\mathcal{H}_0$ and $\mathcal{H}_1$ of the indicated variables. Let us now require the absence of third derivatives in $\left. \nabla^a \nabla_e P_{ab}{}^{ed} \right \vert_{N,f}$. Using \eqref{eq:lagnfproceso} in \eqref{eq:nnPSSS}, one observes that the potentially dangerous terms arise now from $\dfrac{\mathrm{d}^2 P_{N,f}^{(3)}}{\mathrm{d}r^2}$, which reads:
\begin{equation}
\dfrac{\mathrm{d}^2 P_{N,f}^{(3)}}{\mathrm{d}r^2}=\frac{4}{(D-2)} \left ( \frac{\partial^2 \mathcal{L}_{N,f}}{\partial  \mathcal{R}_{N,f}^{(2)}  \mathcal{R}_{N,f}^{(3)}} \dfrac{\mathrm{d}^2 \mathcal{R}_{N,f}^{(2)}}{\mathrm{d}r^2} +\frac{\partial^2 \mathcal{L}_{N,f}}{\partial \left ( \mathcal{R}_{N,f}^{(3)}\right)^2 } \dfrac{\mathrm{d}^2 \mathcal{R}_{N,f}^{(3)}}{\mathrm{d}r^2}-\frac{1}{r}\frac{\partial^2 \mathcal{L}_{N,f}}{\partial  \mathcal{R}_{N,f}^{(1)}  \mathcal{R}_{N,f}^{(4)}} \dfrac{\mathrm{d} \mathcal{R}_{N,f}^{(1)}}{\mathrm{d}r}  \right)\,.
\end{equation}
Imposing the latter to be safe from third derivatives of $f$ and $N$, one gets the conditions:
\begin{equation}
\frac{\partial^2 \mathcal{L}_{N,f}}{\partial \left ( \mathcal{R}_{N,f}^{(3)}\right)^2 }=0\,, \quad \frac{\partial^2 \mathcal{L}_{N,f}}{\partial  \mathcal{R}_{N,f}^{(2)}  \mathcal{R}_{N,f}^{(3)}}=\frac{\partial^2 \mathcal{L}_{N,f}}{\partial  \mathcal{R}_{N,f}^{(1)}  \mathcal{R}_{N,f}^{(4)}}\,.
\end{equation}
Enforcing these constraints on \eqref{eq:lagnfproceso} and taking into account the symmetric dependence of $\mathcal{L}_{N,f}$ on the variables $\mathcal{R}_{N,f}^{(2)}$ and $\mathcal{R}_{N,f}^{(3)}$:
\begin{equation}
\mathcal{L}_{N,f}=\mathcal{H}_0\left (\mathcal{R}_{N,f}^{(4)} \right) \mathcal{R}_{N,f}^{(1)}+\mathcal{H}_1\left (\mathcal{R}_{N,f}^{(4)}\right) \mathcal{R}_{N,f}^{(2)}\mathcal{R}_{N,f}^{(3)}+ \left ( \mathcal{R}_{N,f}^{(2)}+\mathcal{R}_{N,f}^{(3)} \right) \mathcal{H}_2 \left (\mathcal{R}_{N,f}^{(4)}\right)+\mathcal{H}_3 \left (\mathcal{R}_{N,f}^{(4)}\right)\,,
\label{eq:lagnfproceso2}
\end{equation}
where $\mathcal{H}_2$ and $\mathcal{H}_3$ are arbitrary functions. Requiring that \eqref{eq:lagnfproceso2} equals the case $N=1$ in \eqref{eq:laghf} (see Proposition \ref{prop2}), the arbitrary functions $\mathcal{H}_K$ for $K=0,1,2,3$ get completely fixed, obtaining automatically \eqref{eq:lagnftot}. 
\end{proof}
\begin{proposition}
Let $\mathcal{L}(g^{ab}, R_{cdef})$ be a higher-curvature theory built from arbitrary polynomial  contractions of curvature tensors. The following are equivalent:
\begin{enumerate}
\item The entropy tensor evaluated on \eqref{eq:sss} is divergenceless, $\left. \nabla_a P^{abcd} \right \vert_{N,f}=0$.
\item  The equations of motion for general static and spherically symmetric configurations \eqref{eq:sss} are second order.
\item The Lagrangian evaluated on \eqref{eq:sss}, $\mathcal{L}_{N,f}$, reads:
\begin{align}
\nonumber
\mathcal{L}_{N,f}&= \frac{1}{N r^{D-2}} \frac{\mathrm{d}}{\mathrm{d}r} \left (\frac{(1-f)^{(D-2)/2}}{2}(Nf'+2fN') \int^\psi h'(\kappa) \kappa^{-D/2} \mathrm{d}\kappa \right)\\\label{eq:totlagnf} &+\frac{1}{r^{D-2}} \frac{\mathrm{d}}{\mathrm{d} r} \left ( r^{D-1} h\left (\frac{1-f}{r^2} \right) \right ) \,,
\end{align}
where $h$ is an analytic function of its argument.
\end{enumerate}
\label{prop:lagsss}
\end{proposition}
\begin{proof} 
``\emph{1} $\rightarrow$ \emph{2}''. Trivial.

``\emph{2} $\rightarrow$ \emph{3}''.  Follows from Proposition \ref{prop:ngen}, \eqref{eq:lagnftot} and \eqref{eq:formhf}.

``\emph{3} $\rightarrow$ \emph{1}''. If the Lagrangian evaluated on a general SSS ansatz \eqref{eq:sss} takes the form \eqref{eq:totlagnf}, direct computation shows it is equivalent to \eqref{eq:lagnftot}, which satisfies that $\left. \nabla_a P^{abcd} \right \vert_{N,f}=0$.
\end{proof}

Let us now study the SSS solutions of those theories with second-order equations on these configurations. 

\begin{proposition}
The unique SSS solutions of a theory $\mathcal{L}(g^{ab},R_{cdef})$ built from polynomial  contractions of curvature tensors and having second-order equations on \eqref{eq:sss} are given by: 
\begin{equation}
\label{eq:eomgensss}
h'\left (\psi \right) N'=0\,, \quad  \frac{\mathrm{d}}{\mathrm{d}r} \left ( r^{D-1} h\left (\psi \right) \right)=0\,,
\end{equation}
where $h$ is a theory-dependent function given by \eqref{eq:obtenerh} in Proposition \ref{prop:eomsssn1}.
\end{proposition}
\begin{proof}
On the most general SSS ansatz \eqref{eq:sss}, using \eqref{eq:RformNf} and \eqref{eq:PNf} we have that:
\begin{align}
\nonumber
\left. P_{ae}{}^{cd} R_{cd}{}^{be} \right \vert_{N,f}&=\frac{\mathcal{R}_{N,f}^{(1)} P_{N,f}^{(1)}+(D-2)\mathcal{R}_{N,f}^{(2)} P_{N,f}^{(2)}}{8} \tau_a{}^b+\frac{\mathcal{R}_{N,f}^{(1)} P_{N,f}^{(1)}+(D-2)\mathcal{R}_{N,f}^{(3)} P_{N,f}^{(3)}}{8} \rho_a{}^b \\& + \frac{\mathcal{R}_{N,f}^{(2)} P_{N,f}^{(2)}+\mathcal{R}_{N,f}^{(3)} P_{N,f}^{(3)}+4(D-3) \mathcal{R}_{N,f}^{(4)} P_{N,f}^{(4)}}{8} \sigma_a{}^b\,,
\label{eq:PRNF}
\end{align}
If $\left. \nabla_a P^{abcd} \right \vert_{N,f}=0$, the equations of motion imply that:
\begin{equation}
\left. P_{te}{}^{cd} R_{cd}{}^{te} \right \vert_{N,f}=\left. P_{re}{}^{cd} R_{cd}{}^{re} \right \vert_{N,f}\,.
\end{equation}
Upon use of \eqref{eq:PRNF}, this condition translates to:
\begin{equation}
-\frac{f h'(\psi) N'}{r N}=0\,,
\label{eq:eq1alt}
\end{equation}
Imposing this condition, one finds that the $tt$ and $rr$ components of the equation of motion are equivalent. Taking into account \eqref{eq:formhf}, both components are proportional to:
\begin{equation}
\frac{1}{r^{D-3}} \frac{\mathrm{d}}{\mathrm{d}r} \left ( r^{D-1} h(\psi) \right)=0\,.
\label{eq:feq}
\end{equation}
The angular components of the equations of motion are automatically solved by virtue of the Bianchi identity. Indeed, let $\left. \mathcal{E}_{ab} \right \vert_{N,f}$ be the components of the equations on \eqref{eq:sss}. Clearly, it must take the form:
\begin{equation}
\label{eq:eomstrucsss}
\left. \mathcal{E}_{ab} \right \vert_{N,f}=\mathcal{E}_{\rm tem} \tau_{ab}+\mathcal{E}_{\rm rad} \rho_{ab}+\mathcal{E}_{\rm ang} \sigma_{ab}\,,
\end{equation} 
and the condition $\nabla^a \mathcal{E}_{ab}=0$ translates into the vanishing of:
\begin{equation}
\frac{\mathrm{d}\mathcal{E}_{\rm rad}}{\mathrm{d}r}+\frac{(D-2) (\mathcal{E}_{\rm rad}-\mathcal{E}_{\rm ang})}{r}+\frac{(2N'f+N f'))}{2N}(\mathcal{E}_{\rm rad}-\mathcal{E}_{\rm tem})=0\,.
\label{eq:biansss}
\end{equation}
As a consequence, $\mathcal{E}_{\rm tem}=\mathcal{E}_{\rm rad}=0$ implies $\mathcal{E}_{\rm ang}=0$. The fact that the function $h(\psi)$ is given by \eqref{eq:obtenerh} follows by evaluating the Lagrangian in the special case $N=1$ and we conclude. 

\end{proof}

\begin{remark}
\label{rem:csp}
Consider the equations of motion \eqref{eq:eomgensss}. If $h(\psi)$ is such that there does not exist $\psi_0 \in \mathbb{R}$ satisfying\footnote{This would be the condition to avoid the existence of degenerate vacua.} $h(\psi_0)=h'(\psi_0)=0$, the unique solution (up to innocent time reparametrizations) to \eqref{eq:eomgensss} is given by:
\begin{equation}
\label{eq:eomffinalnf}
h(\psi)=\frac{2\mathsf{M}}{r^{D-1}}\,, \quad N=1\,.
\end{equation}
However, if there exists $\psi_0 \in \mathbb{R}$ such that $h(\psi_0)=h'(\psi_0)=0$, in addition to \eqref{eq:eomffinalnf}, any configuration of the form:
\begin{equation}
f(r)=1-\psi_0 r^2\,, \quad N(r) \,\, \rm{ arbitrary}
\end{equation}
also solves the equations of motion. Theories satisfying this property form a zero-measure set in the space of QTGs of type II. Within this particular subset one may find precisely those special combinations of Lovelock gravities which admit a Chern-Simons (odd $D$) or a Born-Infeld formulation (even $D$) \cite{Troncoso:1999pk,Crisostomo:2000bb,Charmousis:2002rc,Zegers:2005vx,Zanelli:2012px}. 

\end{remark}

\begin{remark}
\label{rm:counterex}
On the other hand, it is easy to see that the set of Notions \ref{def:1} and \ref{def:2} (QTGs of type I) is not equivalent to the set of Notions \ref{def:3}, \ref{def:4} and \ref{def:5} (QTGs of type II). While the latter clearly imply the former, let us present an example of a QTG of type I which is not a QTG of type II. Consider the theory:
\begin{align}
\nonumber
\hat{\mathcal{Z}}_{(4)}&=R^4+\frac{6 D (D-1)  R^2
 W^{a b c d} W_{a b c d}}{(D-2) (D-3)}-\frac{24 D(D-1)  R^2
   Z^{ab} Z_{ab}}{(D-2)^2} +\frac{64 D^2(D-1)^2 R
   Z_a^b Z_b^cZ^a_c}{(D-2)^4}\\\nonumber&+\frac{96 D^2 (D-1)^2 R  W\indices{_a_c^b^d}Z^a_b Z^c_d}{
   (D-2)^3 (D-3)}-\frac{96 D^2 (D-1)^2  R
   W_{a c d e}W^{bcde}  Z^a_b }{(D-3) (D-4) (D-2)^2}\\\nonumber&+\frac{8 D^2(D-1)^2 (2
   D-3) R W\indices{^a^b_c_d}W\indices{^c^d_e_f}W\indices{^e^f_a_b}}{(D-2) (D-3) (D ((D-9)
   D+26)-22)}+\frac{48 D^3 (D-1)^3
  \left(Z_a^b Z^a_b\right)^2}{(D-2)^5 (D-3)} \\\nonumber&+\frac{3 D^3(D-1)^2 (3 D-4) \left(W^{a b c d} W_{a b c d}\right)^2}{(D-2)^4
   (D-3)^2}- \frac{384 D^3 (D-1)^3  Z_{a c} Z_{d e} W^{bdce} Z^{a}_b}{(D-3)
   (D-4) (D-2)^4}\\\nonumber & -\frac{96 D^3 (D-1)^3 
  Z_a^bZ_b^cZ_c^dZ_d^a}{(D-2)^5(D-3) } -\frac{48 D^2 (D-1)^4 W_{a b c d} W^{a b c d}Z_{ef} Z^{ef}}{(D-3)
   (D-2)^4}\\ \nonumber &-\frac{192 D^2 (D-1)^3  W_{acbd} W^{c efg} W^d{}_{efg} Z^{ab}}{
   \left(D^2-5 D+6\right)^2 (D-4)}+\frac{96 D^2(2D-1)(D-1)^3 W_{abcd} W^{aecf} Z^{bd} Z_{ef}}{(D-2)^3(D-3)^2}\\&+\beta \left( Z_a^bZ_b^cZ_c^dZ_d^a-\frac{(D^2-6D+12)\left(Z_a^b Z^a_b\right)^2}{2D(D-2)} \right)\,,
\end{align}
where $\beta$ is an arbitrary coupling. Let $\hat{P}^{(4)}_{abcd}$ be the entropy tensor associated with $\hat{\mathcal{Z}}_{(4)}$. Direct computation reveals:
\begin{equation}
\left. \nabla^d \hat{P}^{(4)}_{abcd} \right \vert_f=0\,,
\end{equation}
when considered on the static spherically symmetric ansatz \eqref{eq:sssn1} with $N=1$. However, for generic $N=N(r)$, it turns out that 
\begin{equation}
\left. \nabla^d \hat{P}^{(4)}_{abcd} \right \vert_{N,f}=\beta \left[ \mathcal{A}_{N,f}\,\delta_{[a}^r \tau_{b]}^d+\mathcal{B}_{N,f}\, \delta_{[a}^r \sigma_{b]}^d \right] \,,
\end{equation}
where $\mathcal{A}_{N,f}$ and $\mathcal{B}_{N,f}$ are certain intricate functions of $r,f,f',f'',f''',N,N',N''$ and $N'''$. In particular, $\left. \nabla^d \hat{P}^{(4)}_{abcd} \right \vert_{N,f}$ vanishes if and only if $\beta=0$, thus showing the existence of theories fulfilling Notions \ref{def:1} and \ref{def:2} but not satisfying Notions \ref{def:3}, \ref{def:4} and \ref{def:5}. 

\end{remark}

\section{Theories with $2^{\rm nd}$-order equations on general SS backgrounds}\label{sec4}

Let us now consider general SS configurations, without assuming the staticity condition:
\begin{equation}
\mathrm{d}s_{\rm SS}^2=-N(t,r)^2f(t,r) \mathrm{d} t^2+\frac{\mathrm{d} r^2}{f(t,r)}+r^2 \mathrm{d} \Omega_{D-2}^2\,.
\label{eq:ss}
\end{equation}
The Riemann curvature tensor of \eqref{eq:ss} adopts the following form:
\begin{align}
\nonumber
\left. R_{ab}{}^{cd} \right \vert_{\rm SS}&=\mathcal{R}_{\rm SS}^{(1)} \tau_{[a}^{[c}  \rho_{b]}^{d]}+\mathcal{R}_{\rm SS}^{(2)} \tau_{[a}^{[c}  \sigma_{b]}^{d]} +\mathcal{R}_{\rm SS}^{(3)} \rho_{[a}^{[c} \sigma_{b]}^{d]}+\mathcal{R}_{\rm SS}^{(4)} \zeta_{[a}^{[c} \sigma_{b]}^{d]} +\mathcal{R}_{\rm SS}^{(5)} \sigma_{[a}^{[c}  \sigma_{b]}^{d]}\,,
\label{eq:Rformss}
\end{align}
where we have defined $\zeta_{ab}= 2 N \delta_{(a}^t \delta_{b)}^r$ and:
\begin{align}
\mathcal{R}_{\rm SS}^{(1)}&=\frac{4 N \left ( \partial_t f \right)^2+ 2 f \partial_t f \partial_t N-2 f N \partial_t^2 f}{f^3 N^3}-\frac{6 \partial_r f \partial_r N+2N \partial_r^2 f+4f \partial_r^2 N}{N}\,, \\ \mathcal{R}_{\rm SS}^{(2)}&=-\frac{2N \partial_r f+4f \partial_r N}{r N}\,, \quad \mathcal{R}_{\rm SS}^{(3)}= -\frac{2 \partial_r f}{r}\,, \quad \mathcal{R}_{\rm SS}^{(4)}= - 2 \frac{\partial_t f}{r N f}\,, \quad \mathcal{R}_{\rm SS}^{(5)}=2 \frac{(1-f)}{r^2}\,.
\label{eq:aesrform}
\end{align}
The following properties hold:
\begin{equation}
\tau_a^b \zeta_{bc}= N \delta_a^t \delta_c^r\,, \quad \rho_a^b \zeta_{bc}= N \delta_a^r \delta_c^t\,, \quad \sigma_a^b \zeta_b^c=0\,, \quad \zeta_{a}^b \zeta_b^c=-\tau_a^c-\rho_a^c\,.
\end{equation}
It will be convenient to present the Weyl curvature tensor $\left. W_{ab}{}^{cd} \right \vert_{\rm SS}$, the traceless Ricci tensor $\left. Z_{ab}\right \vert_{\rm SS}$ and the Ricci scalar $\left. R \right \vert_{\rm SS}$ on \eqref{eq:ss}. These take the form:
\begin{align}
\label{eq:weylss}
&\left.W_{ab}{}^{cd}\right\vert_{\rm SS}=\Omega_{\rm SS} \left[ (D-2)(D-3) \tau_{[a}^{[c} \rho_{b]}^{d]}-(D-3)\left (\tau_{[a}^{[c}+\rho_{[a}^{[c} \right ) \sigma_{b]}^{d]}+ \sigma_{[a}^{[c} \sigma_{b]}^{d]} \right]\, , \\ 
\label{eq:zzss}
&\left.Z_a^b\right\vert_{\rm SS}=\frac{(D-2)\mathcal{R}_{\rm SS}^{(4)}}{4} \zeta_{a}{}^{b}-\Theta_{\rm SS}\left[ \frac{D-2}{2} \left (\tau_a^b+\rho_a^b \right)- \sigma_a^b \right ]+\frac{(D-2)\left (\mathcal{R}_{\rm SS}^{(2)}-\mathcal{R}_{\rm SS}^{(3)} \right )}{8} \left[ \tau_a^b-\rho_a^b\right]\, , \\ &\left. R \right \vert_{\rm SS}= \frac{\mathcal{R}_{\rm SS}^{(1)}+(D-2)\left (\mathcal{R}_{\rm SS}^{(2)}+\mathcal{R}_{\rm SS}^{(3)} \right)+(D-2)(D-3)\mathcal{R}_{\rm SS}^{(5)}}{2} \,.
\label{eq:ricss}
\end{align}
where
\begin{align}
\label{eq:omss}
\Omega_{\rm SS}=\frac{\mathcal{R}_{\rm SS}^{(1)}-\mathcal{R}_{\rm SS}^{(2)}-\mathcal{R}_{\rm SS}^{(3)}+2\mathcal{R}_{\rm SS}^{(5)}}{(D-1)(D-2)}&\, , \quad \Theta_{\rm SS}=\frac{-2\mathcal{R}_{\rm SS}^{(1)}-(D-4)\left (\mathcal{R}_{\rm SS}^{(2)}+\mathcal{R}_{\rm SS}^{(3)} \right )+4(D-3)\mathcal{R}_{\rm SS}^{(5)}}{4 D}\,.
\end{align}
\begin{proposition}
\label{prop:ssfoem}
Let $\mathcal{L}(g^{ab},R_{cdef})$ be a higher-curvature gravity constructed from polynomial contractions of curvature tensors and with second-order equations on static and spherically symmetric backgrounds. Then, evaluated on a general (non-static) spherically symmetric background \eqref{eq:ss}, its expression reads:
\begin{align}
\nonumber
\mathcal{L}_{\rm SS}&=h_0(\psi) \left (\partial_r^2 f-\frac{\partial_t f \partial_t N}{f^2 N^3}+ \frac{2f}{N}\partial_r^2 N+\frac{\partial_t^2 f}{f^2 N^2}+\frac{3 \partial_r f \partial_r N}{N}-\frac{2(\partial_t f)^2}{N^2 f^3} \right) +\frac{\partial h_0(\psi)}{\partial f}(\partial_r f)^2\\ &+h_1(\psi)\frac{\partial_r f}{r}+h_2(\psi) +\frac{2 f\partial_r f \partial_r N}{N}\frac{\partial h_0(\psi)}{\partial f} +\frac{f \partial_r N}{r N}h_1(\psi)+\frac{\partial h_0(\psi)}{\partial f} \frac{(\partial_t f )^2}{N^2 f^2} \,,
\label{eq:lagssform}
\end{align}
where $\psi=\dfrac{1-f(t,r)}{r^2}$ and $h_0$, $h_1$ and $h_2$ were defined back in Proposition \ref{prop2}.
\end{proposition}
\begin{proof}
Observe that the traceless Ricci tensor $Z_{ab}$ on the general (non-static) SS configuration \eqref{eq:ss} may be written as:
\begin{equation}
\left.Z_{ab}\right\vert_{\rm SS}=\Sigma_{ab}+\Theta_{\rm SS}\, \sigma_{ab}\,,
\end{equation}
where $\Sigma_{ab}$ is a rank-two tensor transverse to the angular coordinates whose expression may be directly obtained by comparing to \eqref{eq:zzss}. As such, $(\tau_{ab}+\rho_{ab}) \Sigma^{bc}=\Sigma_{a}{}^c$ and $\Sigma_{ab} \sigma^{bc}=0$. Noting that $(\tau_{[a}{}^{[c}+\rho_{[a}{}^{[c})(\tau_{b]}{}^{d]}+\rho_{b]}{}^{d]})=2\tau_{[a}^{[c} \rho_{b]}^{d]}$, we learn that any curvature invariant on \eqref{eq:ss} will be entirely expressed in terms of $\Omega_{\rm SS},\Theta_{\rm SS}, \left. R \right \vert_{\rm SS}$ and scalars formed by the arbitrary contraction of $\Sigma_{ab}$ tensors with projectors $\tau_{ab}+\rho_{ab}$. However, one may show\footnote{This can be proved by direct use of \emph{Schouten} identities --- \ie considering a given tensor product $\Sigma_{a_1}{}^{b_1 } \dots  \Sigma_{a_n}{}^{b_n}$, antisymmetrizing over more than two indices and setting the resulting expression to zero. } that any such scalar will be solely a function of two invariants: $\Sigma_{a}^{a}$ and $\Sigma_{ab} \Sigma^{ab}$, which read as follows:
\begin{equation}
\Sigma_{a}^{a}{}=-(D-2) \Theta_{\rm SS}\, \quad  \Sigma_{ab} \Sigma^{ab}=\frac{(D-2)^2}{32}\left ( \left (\mathcal{R}_{\rm SS}^{(2)}-\mathcal{R}_{\rm SS}^{(3)} \right)^2 -4 \left (\mathcal{R}_{\rm SS}^{(4)} \right)^2+16 \Theta_{\rm SS}^2 \right)\,.
\end{equation}
Therefore, any curvature invariant on \eqref{eq:ss} will be a function of $\Omega_{\rm SS}, \Theta_{\rm SS}, \Sigma_{ab} \Sigma^{ab}$ and $\left. R \right \vert_{SS}$. Interestingly enough, $\Omega_{\rm SS}, \Theta_{\rm SS}$ and $\left. R \right \vert_{\rm SS}$ do not depend on $\mathcal{R}_{\rm SS}^{(2)}-\mathcal{R}_{\rm SS}^{(3)}$: only $\Sigma_{ab} \Sigma^{ab}$ does (in addition, quadratically). Furthermore, the only difference between the static and non-static cases consists in the appearance of $\left (\mathcal{R}_{\rm SS}^{(4)} \right)^2$, which appears only in $\Sigma_{ab} \Sigma^{ab}$ as well.

This suggests a simple way to obtain the Lagrangian on a general SS background from its expression on a SSS one. Specifically, identify in the SSS Lagrangian $\mathcal{L}_{N,f}$ the whole dependence on $\left (\mathcal{R}_{N,f}^{(2)}-\mathcal{R}_{N,f}^{(3)} \right)^2$ --- which is in fact formally equivalent to  $\left ( \mathcal{R}_{\rm SS}^{(2)}-\mathcal{R}_{\rm SS}^{(3)} \right)^2 $ ---, and perform the formal replacement:
\begin{align}
\nonumber
\left\lbrace \mathcal{R}_{N,f}^{(1)},\mathcal{R}_{N,f}^{(2)}+\mathcal{R}_{N,f}^{(3)},\left ( \mathcal{R}_{N,f}^{(2)}-\mathcal{R}_{N,f}^{(3)} \right)^2,\mathcal{R}_{N,f}^{(4)}  \right\rbrace & \rightarrow \\& \hspace{-4cm} \left\lbrace \mathcal{R}_{\rm SS}^{(1)},\mathcal{R}_{\rm SS}^{(2)}+\mathcal{R}_{\rm SS}^{(3)},\left ( \mathcal{R}_{\rm SS}^{(2)}-\mathcal{R}_{\rm SS}^{(3)} \right)^2 -4 \left ( \mathcal{R}_{\rm SS}^{(4)}\right)^2 ,\mathcal{R}_{\rm SS}^{(5)}\right\rbrace\,.
\end{align}
The most convenient form of $\mathcal{L}_{N,f}$ in which to perform this formal replacement is \eqref{eq:lagnfproceso2}. By doing this, one gets:
\begin{align}
\mathcal{L}_{\rm SS}&=\mathcal{H}_0\left (\mathcal{R}_{\rm SS}^{(5)} \right) \mathcal{R}_{\rm SS}^{(1)}+ \left ( \mathcal{R}_{\rm SS}^{(2)}+\mathcal{R}_{\rm SS}^{(3)} \right) \mathcal{H}_2 \left (\mathcal{R}_{\rm SS }^{(5)}\right)+\mathcal{H}_3 \left (\mathcal{R}_{\rm SS}^{(5)}\right)\\& +\mathcal{H}_1\left (\mathcal{R}_{\rm SS}^{(5)}\right)\left [ \frac{\left (\mathcal{R}_{\rm SS}^{(2)}+\mathcal{R}_{\rm SS}^{(3)} \right)^2-\left (\mathcal{R}_{\rm SS}^{(2)}-\mathcal{R}_{\rm SS}^{(3)} \right)^2}{4}+\left ( \mathcal{R}_{\rm SS}^{(4)}\right)^2 \right] \,,
\end{align}
where the functions $\mathcal{H}_K$ with $K=0,1,2,3$ are fixed by imposing $\mathcal{L}_{\rm SS}$ adopts the form in \eqref{eq:laghf} for the restricted case $N=1$ and $f=f(r)$. Doing this, one gets \eqref{eq:lagssform} and we conclude.

\end{proof}

\begin{remark}
The Lagrangian \eqref{eq:lagssform} may be equivalently written as follows:
\begin{equation}
    \mathcal{L}_{\rm SS}=-\frac{h_0(\psi)}{2} \mathcal{R}_{\rm SS}^{(1)}-\frac{h_0'(\psi)}{4} \left ( \mathcal{R}_{\rm SS}^{(2)} \mathcal{R}_{\rm SS}^{(3)}+\left ( \mathcal{R}_{\rm SS}^{(4)}\right)^2 \right)-\frac{h_1(\psi)}{4} \left ( \mathcal{R}_{\rm SS}^{(2)}+\mathcal{R}_{\rm SS}^{(3)} \right) +h_2(\psi)\,,
\end{equation}
for $\psi=\dfrac{\mathcal{R}_{\rm SS}^{(5)}}{2}=\dfrac{1-f}{r^2}$ and where the functions $h_0,h_1$ and $h_2$ were defined in \eqref{eq:formhf}.
\end{remark}

Before continuing, let us express the entropy tensor $P_{abcd}$ in terms of derivatives of the Lagrangian with respect to the functions (and derivatives thereof) appearing in the metric ansatz \eqref{eq:ss}. The structure of the Riemann curvature tensor \eqref{eq:Rformss} ensures that:
\begin{align}
\left. P_{ab}{}^{cd} \right \vert_{\rm SS}&=P^{(1)}_{\rm SS} \tau_{[a}^{[c}  \rho_{b]}^{d]}+P^{(2)}_{\rm SS} \tau_{[a}^{[c}  \sigma_{b]}^{d]} +P^{(3)}_{\rm SS} \rho_{[a}^{[c} \sigma_{b]}^{d]}+P^{(4)}_{\rm SS} \zeta_{[a}^{[c} \sigma_{b]}^{d]} +P^{(5)}_{\rm SS} \sigma_{[a}^{[c}  \sigma_{b]}^{d]}\,,
\label{eq:Pformss}
\end{align}
where 
\begin{align}
\notag 
P_{\rm SS}^{(1)}&=4\frac{\partial \mathcal{L}_{\rm SS}}{\partial \mathcal{R}_{\rm SS}^{(1)}}\,,  \quad  P_{\rm SS}^{(I)}=\frac{4}{D-2}\frac{\partial \mathcal{L}_{\rm SS}}{\partial \mathcal{R}_{\rm SS}^{(I)}}\,, \quad I=2,3 \,, \\P_{\rm SS}^{(4)}&=-\frac{2}{(D-2)}\frac{\partial \mathcal{L}_{\rm SS}}{\partial \mathcal{R}_{\rm SS}^{(4)}}\,,\quad
P_{\rm SS}^{(5)}=\frac{2}{(D-2)(D-3)}\frac{\partial \mathcal{L}_{\rm SS}}{\partial \mathcal{R}_{\rm SS}^{(5)}}\,.\quad 
\label{eq:pisssfacil}
\end{align}
Equivalently, in terms of the derivatives of $f$ and $N$:
\begin{align}
\label{eq:P1ss}
P^{(1)}_{\rm SS}&=-2\frac{\partial \mathcal{L}_{\rm SS}}{\partial (\partial_r^2 f)}\,, \quad P^{(2)}_{\rm SS}=-\frac{r}{(D-2)f}\left[ N\frac{\partial \mathcal{L}_{\rm SS}}{\partial (\partial_r N)}-3 \partial_r f \frac{\partial \mathcal{L}_{\rm SS}}{\partial (\partial_r^2 f)} \right]\,, \\  \label{eq:P3ss}
P^{(3)}_{\rm SS}&=-P^{(2)}_{\rm SS}-\frac{2r}{(D-2)} \left[\frac{\partial \mathcal{L}_{\rm SS}}{\partial (\partial_r f)}-\frac{3 \partial_r N}{N} \frac{\partial \mathcal{L}_{\rm SS}}{(\partial_r^2 f)} \right]\,, \\
\label{eq:P4ss}
P^{(4)}_{\rm SS}&=\frac{r}{(D-2) f^2 N^2} \left[f^3 N^3 \frac{\partial \mathcal{L}_{\rm SS}}{\partial (\partial_t f)}+ \left( 4N \partial_t f+f \partial_t N \right) \frac{\partial \mathcal{L}_{\rm SS}}{\partial (\partial_r^2 f)}  \right]
 \\  
 \label{eq:P5ss}
P^{(5)}_{\rm SS}&=-\frac{r \partial_r N}{(D-3) N} P^{(2)}_{\rm SS} -\frac{r^2}{(D-2)(D-3)}\left[\frac{\partial \mathcal{L}_{\rm SS}}{\partial f}-\frac{2 \partial_r^2 N}{N} \frac{\partial \mathcal{L}_{\rm SS}}{\partial (\partial_r^2 f)} \right]\\\nonumber &+ \frac{r^2}{(D-2)(D-3)} \left[ -\frac{\partial_t f}{f}\frac{\partial \mathcal{L}_{\rm SS}}{\partial (\partial_t f)}+ \frac{1}{N^3 f^4}\left ( 2 N (\partial_t f)^2+f \partial_t f \partial_t N-2f N \partial_t^2 f  \right) \frac{\partial \mathcal{L}_{\rm SS}}{\partial (\partial_r^2 f)} \right]\,.
\end{align}

Let us note that $P^{(4)}_{\rm SS}$ may be obtained from $P^{(2)}_{\rm SS}$ and $P^{(3)}_{\rm SS}$ if one notes that the Lagrangian on \eqref{eq:ss} depends on $\mathcal{R}_{\rm SS}^{(2)}, \mathcal{R}_{\rm SS}^{(3)} $  and $\mathcal{R}_{\rm SS}^{(4)}$ through the combinations $\mathcal{R}_{\rm SS}^{(2)}+\mathcal{R}_{\rm SS}^{(3)} $ and $\left (\mathcal{R}_{\rm SS}^{(2)}-\mathcal{R}_{\rm SS}^{(3)} \right)^2 -4 \left (\mathcal{R}_{\rm SS}^{(4)} \right)^2$ --- see proof of Proposition \ref{prop:ssfoem}. Indeed, one gets:
\begin{equation}
  \mathcal{R}_{\rm SS}^{(4)} P^{(3)}_{\rm SS}+\mathcal{R}_{\rm SS}^{(2)}P^{(4)}_{\rm SS}=\mathcal{R}_{\rm SS}^{(3)}P^{(4)}_{\rm SS}+\mathcal{R}_{\rm SS}^{(4)}P^{(2)}_{\rm SS}\,.
    \label{eq:propcuriosa}
\end{equation}

Now, the divergence of \eqref{eq:Pformss} takes the following form:
\begin{equation}
\left. \nabla_c P_{ab}{}^{cd} \right \vert_{\rm SS}=A_{\rm SS} \delta_{[a}^t \rho_{b]}^d+B_{\rm SS} \delta_{[a}^r \tau_{b]}^d+ C_{\rm SS} \delta_{[a}^t \sigma_{b]}^d+E_{\rm SS}  \delta_{[a}^r \sigma_{b]}^d\,,
\label{eq:nablaPss}
\end{equation}
where we have defined:
\begin{align}
\label{eq:ABss}
A_{\rm SS}&=\frac{\partial_t P^{(1)}_{\rm SS}}{2}+\frac{Nf(D-2)P^{(4)}_{\rm SS}}{2r} \,, \quad B_{\rm SS}=\frac{\partial_r P^{(1)}_{\rm SS}}{2}+ \frac{(D-2)(P^{(1)}_{\rm SS}-P^{(2)}_{\rm SS})}{2r}\,, \\
\label{eq:Css}
C_{\rm SS}&=\frac{\partial_t P^{(2)}_{\rm SS}}{2}+\frac{Nf \partial_r P^{(4)}_{\rm SS}}{2}- \frac{\partial_t f \left (P^{(2)}_{\rm SS}-P^{(3)}_{\rm SS} \right )}{4f} +\frac{P^{(4)}_{\rm SS}(2r \partial_r N f+r N \partial_r f+(D-3) Nf)}{2r}\,,\\
\nonumber
E_{\rm SS}&=\frac{\partial_r P^{(3)}_{\rm SS}}{2}-\frac{\partial_t P^{(4)}_{\rm SS}}{2Nf}-\frac{\left (P^{(2)}_{\rm SS}-P^{(3)}_{\rm SS} \right )(2 \partial_r N f+ N \partial_r f)}{4Nf}+ \frac{(D-3)\left (P^{(3)}_{\rm SS}-2P^{(5)}_{\rm SS} \right )}{2r}\\ \label{eq:Ess} &+\frac{\partial_t f P^{(4)}_{\rm SS}}{2N f^2}\,.
\end{align}

\begin{proposition}
\label{prop:ssres}
Let $\mathcal{L}(g^{ab},R_{cdef})$ be a higher-curvature gravity constructed from polynomial contractions of curvature tensors. The following are equivalent.
\begin{enumerate}
\item $\mathcal{L}(g^{ab},R_{cdef})$ has second-order equations of motion on static and spherically symmetric backgrounds.
\item $\left. \nabla_a P^{abcd} \right \vert_{\rm SS}=0$.
\item $\mathcal{L}(g^{ab},R_{cdef})$ has second-order equations of motion on spherically symmetric backgrounds.
\item The Lagrangian evaluated on \eqref{eq:ss} is given by:
\begin{align}
\label{eq:totlagss}
\mathcal{L}_{\rm SS}&= \frac{1}{N r^{D-2}} \frac{\mathrm{d}}{\mathrm{d}r} \left (\frac{(1-f)^{(D-2)/2}}{2}(N\partial_r f+2f\partial_r N) \int^\psi h'(\kappa) \kappa^{-D/2} \mathrm{d}\kappa \right)\\ \nonumber &+\frac{1}{r^{D-2}} \frac{\mathrm{d}}{\mathrm{d} r} \left ( r^{D-1} h\left (\frac{1-f}{r^2} \right) \right )+\frac{1}{N} \frac{\partial}{\partial t} \left (\frac{\partial_t f}{2N f^2} \left (\frac{1-f}{r^2} \right)^{(D-2)/2}\int^\psi h'(\kappa) \kappa^{-D/2} \mathrm{d}\kappa  \right) \,,
\end{align}
where $h$ is some analytic function.
\end{enumerate}
\end{proposition}

\begin{proof}
``\emph{1} $\rightarrow$ \emph{2}''. If the theory possesses second-order equations on SSS configurations \eqref{eq:sss}, then on non-static SS backgrounds \eqref{eq:ss} the Lagrangian will take the form prescribed in \eqref{eq:lagssform}, as proven by Proposition \ref{prop:ssfoem}. Taking this expression, we find that the subsequent components of \eqref{eq:Pformss} take the form:
\begin{align}
\nonumber
P^{(1)}_{\rm SS}&=-2h_0(\psi)\,, \quad  P^{(2)}_{\rm SS}=\frac{2h_0'(\psi) \partial_r f-r h_1(\psi)}{(D-2)r}\,,  \\ \label{eq:Pss}P^{(3)}_{\rm SS}&=P^{(2)}_{\rm SS}+\frac{4f h_0'(\psi)\partial_r N}{(D-2) r N}\,, \quad P^{(4)}_{\rm SS}=-\frac{2h_0'(\psi) \partial_t f}{(D-2)r fN}\,, \\ \nonumber P^{(5)}_{\rm SS}&=\frac{P^{(3)}_{\rm SS}}{2}+ \frac{r}{2(D-3)} \left ( \partial_r P^{(3)}_{\rm SS}-\frac{\partial_t \left (\frac{ P^{(4)}_{\rm SS}}{f} \right)}{N}  +\frac{r \mathcal{R}_{\rm SS}^{(2)} \left ( P^{(2)}_{\rm SS}-P^{(3)}_{\rm SS}\right)}{4f}\right)\,,
\end{align} 
where $r^2 \psi=1-f$, the functions $h_0$ and $h_1$ were given in \eqref{eq:formhf}. From here, it is a straightforward exercise to check that $A_{\rm SS},B_{\rm SS},C_{\rm SS}$ and $E_{\rm SS}$ defined in \eqref{eq:ABss}, \eqref{eq:Css} and \eqref{eq:Ess} are identically zero, so that $\left. \nabla_a P^{abcd} \right \vert_{\rm SS}=0$.
\\

``\emph{2} $\rightarrow$ \emph{3}''. Trivial.  \\

``\emph{3} $\rightarrow$ \emph{4}''. If $\mathcal{L}(g^{ab},R_{cdef})$ has second-order equations on \eqref{eq:ss}, then its Lagrangian is given by \eqref{eq:lagssform}. It is then a direct computation to show that it can be written as in \eqref{eq:totlagss}.    \\

``\emph{4} $\rightarrow$ \emph{1}''. By direct computation, one notes that the Lagrangian \eqref{eq:totlagss} satisfies $\left. \nabla_a P^{abcd} \right \vert_{\rm SS}=0$. Therefore, the theory has second-order equations on spherical backgrounds.

\end{proof}

At this point, it is convenient to examine the conditions the contracted Bianchi identity imposes on the gravitational equations of motion $\mathcal{E}_{ab}$ of a theory $\mathcal{L}(g^{ab},R_{cdef})$ when considered on spherical backgrounds \eqref{eq:ss}. Spherical symmetry enforces that:
\begin{equation}
\label{eq:eomstrucss}
\left. \mathcal{E}_{ab} \right \vert_{\rm SS}=\mathcal{E}_{\rm tem} \tau_{ab}+\mathcal{E}_{\rm mix} \zeta_{ab}+\mathcal{E}_{\rm rad} \rho_{ab}+\mathcal{E}_{\rm ang} \sigma_{ab}\,.
\end{equation}
Now, the contracted Bianchi identity $\left. \nabla^a \mathcal{E}_{ab} \right \vert_{\rm SS}=0$ imposes the following relations on the different components of $\left. \mathcal{E}_{ab} \right \vert_{\rm SS}$:
\begin{align}
\label{eq:B1}
&\partial_t \mathcal{E}_{\rm tem} +\frac{(\mathcal{E}_{\rm rad}-\mathcal{E}_{\rm tem})\partial_t f}{2f} +\left (N \partial_r f +2f \partial_r N +\frac{(D-2)Nf}{r} \right)\mathcal{E}_{\rm mix}+Nf \partial_r \mathcal{E}_{\rm mix}=0\,, \\ \nonumber
&2f^2(r(\mathcal{E}_{\rm rad}-\mathcal{E}_{\rm tem})\partial_r N+N(D-2)(\mathcal{E}_{\rm rad}-\mathcal{E}_{\rm ang})+ r N \partial_r \mathcal{E}_{\rm rad})+2 r \mathcal{E}_{\rm mix} \partial_t f\\& +r f N (\mathcal{E}_{\rm rad}-\mathcal{E}_{\rm tem})\partial_r f-2r f \partial_t \mathcal{E}_{\rm mix}=0\,. \label{eq:B2}
\end{align}

\begin{proposition}
The unique spherically symmetric solutions of a theory $\mathcal{L}(g^{ab},R_{cdef})$ constructed from analytic functions of polynomial curvature invariants with second-order equations of motion on \eqref{eq:ss} are given by:
\begin{equation}
\label{eq:eqsss}
h'(\psi)\partial_t f=0\,, \quad h'(\psi)\partial_r N=0\,, \quad \frac{\mathrm{d}}{\mathrm{d}r} \left(r^{D-1} h(\psi) \right)=0\,.
\end{equation}
where $h$ is a theory-dependent function given by \eqref{eq:obtenerh} in Proposition \ref{prop:eomsssn1} and $\psi=\dfrac{1-f}{r^2}$.
\end{proposition}

\begin{proof}
For general SS backgrounds, we have that
\begin{align}
\nonumber
\left. P_{ae}{}^{cd} R_{cd}{}^{be} \right \vert_{\rm SS}&=\frac{\mathcal{R}_{\rm SS}^{(1)} P^{(1)}_{\rm SS}+(D-2)(\mathcal{R}_{\rm SS}^{(2)} P^{(2)}_{\rm SS}-\mathcal{R}_{\rm SS}^{(4)} P^{(4)}_{\rm SS})}{8} \tau_a{}^b  \\ \nonumber &+ \frac{(D-2)\left (\mathcal{R}_{\rm SS}^{(4)} P^{(2)}_{\rm SS}+\mathcal{R}_{\rm SS}^{(3)} P^{(4)}_{\rm SS} \right)}{8}\zeta_a{}^b\\\nonumber & +\frac{\mathcal{R}_{\rm SS}^{(1)} P^{(1)}_{\rm SS}+(D-2)(\mathcal{R}_{\rm SS}^{(3)} P^{(3)}_{\rm SS}-\mathcal{R}_{\rm SS}^{(4)} P^{(4)}_{\rm SS})}{8} \rho_a{}^b\\&+ \frac{\mathcal{R}_{\rm SS}^{(2)} P^{(2)}_{\rm SS}+\mathcal{R}_{\rm SS}^{(3)} P^{(3)}_{\rm SS}-2\mathcal{R}_{\rm SS}^{(4)}P^{(4)}_{\rm SS}+4(D-3) \mathcal{R}_{\rm SS}^{(5)} P^{(5)}_{\rm SS}}{8} \sigma_a{}^b\,,
\label{eq:PRss}
\end{align}
where we have used \eqref{eq:propcuriosa}. 
Let us examine first the component $\mathcal{E}_{t}{}^{r}$ of the equations of motion. This is entirely given by $\mathcal{R}_{\rm SS}^{(4)} P^{(2)}_{\rm SS}+\mathcal{R}_{\rm SS}^{(3)} P^{(4)}_{\rm SS}=0$, as the theory satisfies $\left. \nabla_a P^{abcd} \right \vert_{\rm SS}=0$. By \eqref{eq:lagssform} and \eqref{eq:Pss}, we compute:
\begin{equation}
\mathcal{E}_{t}{}^{r}=\frac{\partial_t f h_1(\psi)}{4 r }=0\,.
\end{equation}
By \eqref{eq:formhf}, this implies:
\begin{equation}
h'(\psi)\partial_t f=0\,.
\end{equation}
On the other hand, if we compute $\mathcal{E}_{t}^t-\mathcal{E}_{r}^r$, we directly infer that:
\begin{equation}
\mathcal{E}_{t}^t-\mathcal{E}_{r}^r=\frac{(D-2)}{8}\left ( \mathcal{R}_{\rm SS}^{(2)} P^{(2)}_{\rm SS}-\mathcal{R}_{\rm SS}^{(3)} P^{(3)}_{\rm SS}\right)= \frac{f \partial_r N}{2r N} h_1(\psi)\,,
\end{equation}
so that we are led to the equation:
\begin{equation}
h'(\psi) \partial_r N=0\,.
\label{eq:ettrrigual}
\end{equation}
Imposing \eqref{eq:ettrrigual}, $\mathcal{E}_{t}^t$ and $\mathcal{E}_{r}^r$ are equivalent and, by use of \eqref{eq:lagssform}, produce the condition:
\begin{equation}
\frac{\mathrm{d}}{\mathrm{d}r} \left(r^{D-1} h(\psi) \right)=0\,.
\end{equation}
Finally, the angular components of the equations are automatically satisfied by the Bianchi identity \eqref{eq:B2} and we conclude.
\end{proof}

{\noindent \bf Proof of Theorem \ref{thm:2}.} Follows  directly from Proposition \ref{prop:lagsss} and by Proposition \ref{prop:ssres}. \qed

\vspace{0.25cm}

It is noteworthy that the requirement of second-order equations on SSS backgrounds actually forces the equations on general spherical configurations to be second order. Not only that: the $tr$ component $\mathcal{E}_{tr}$ of the equations of motion must take the form $\mathcal{E}_{tr}=-\frac{h'(\psi)}{2rf} \partial_t f$, which implies the time independence of $f$. 

This motivates the converse question: given a theory constructed from analytic combinations of polynomial curvature invariants whose $\mathcal{E}_{tr}$ equation \emph{only} depends on $\partial_t f$, $f$ and $N$ on spherical backgrounds \eqref{eq:ss} with $N=N(r)$ --- to allow for even more general structures for $\mathcal{E}_{tr}$ if there was a general time-dependent $N$ ---,  what could we conclude about the theory? Despite being seemingly a much weaker condition, we prove in the appendix that the theory turns out to be a type II QTG, thus promoting this very special structure of $\mathcal{E}_{tr}$ to be a defining property of type II QTGs.

\section{Birkhoff implies QTG (of type II)}
\label{sec5}

Let $\mathcal{L}(g^{ab},R_{cdef})$ be a higher-curvature theory of gravity which depends analytically on polynomial curvature invariants. Assume it satisfies a Birkhoff theorem, so that  all spherically symmetric solutions
are characterized by a single continuous parameter to be identified with the ADM mass, together with a possible additional discrete parameter. Specifically, if we focus on general SSS metrics \eqref{eq:sss}, this means that the $tt$ and $rr$ components of the equations of motion may not possess derivatives of degree higher than one. Indeed, from these equations, one would get two integration constants: one is to be compensated with the gauge freedom associated with time rescalings, while the other will be identified with the ADM mass. Consequently, higher (radial) derivatives of $f$ and $N$ in the $tt$ or $rr$ components would give rise to unconstrained continuous integration constants that would violate Birkhoff's theorem.

It is direct to see that any theory $\mathcal{L}(g^{ab},R_{cdef})$ whose $tt$ and $rr$ equations on \eqref{eq:sss} are of first order in derivatives will be a type II QTG.
\begin{proposition}
\label{prop:sbprop}
    Let $\mathcal{L}(g^{ab},R_{cdef})$ be a theory of gravity which is built from analytic combinations of curvature invariants. Consider general static and spherically symmetric configurations \eqref{eq:sss}. If the $tt$ and $rr$ components of the equations are of first derivative order, then the theory is a type II QTG.
\end{proposition}
\begin{proof}
    On the one hand, from the structure of the curvature tensor and the tensor $\nabla^{c} \nabla^d P_{acbd}$ on \eqref{eq:sss}, which are given by  \eqref{eq:RformNf} and  \eqref{eq:nnPSSS} respectively, one infers that the equations of motion will have the structure \eqref{eq:eomstrucsss}. In particular, there are only three independent equations: the $tt$ component given by $\mathcal{E}_{\rm tem}$, the $rr$ component determined by $\mathcal{E}_{\rm rad}$ and the angular equation $\mathcal{E}_{\rm ang}$ --- see \eqref{eq:eomstrucsss}. However, the contracted Bianchi identity on \eqref{eq:sss}, which was presented in \eqref{eq:biansss}, imposes that:
    \begin{equation}
       \mathcal{E}_{\rm ang}= \mathcal{E}_{\rm rad}+ \frac{r}{D-2} \left[\frac{\mathrm{d}\mathcal{E}_{\rm rad}}{\mathrm{d}r}+\frac{(2N'f+N f'))}{2N}(\mathcal{E}_{\rm rad}-\mathcal{E}_{\rm tem}) \right]\,.
    \end{equation}
    As a result, if $\mathcal{E}_{\rm tem}$ and $\mathcal{E}_{\rm rad}$ are of first derivative order, $\mathcal{E}_{\rm ang}$ will be of second derivative order. Consequently, the whole set of equations on \eqref{eq:sss} will of second derivative order. Proposition \ref{prop:ssres}  guarantees that $\mathcal{L}(g^{ab},R_{cdef})$ is a type II QTG and we conclude.
\end{proof}

{\noindent \bf Proof of Theorem \ref{thmBirk}.} Assume that a certain theory $\mathcal{L}(g^{ab},R_{cdef})$ built from analytic combinations of polynomial curvature invariants is such that all their SSS solutions are characterized by a single continuous parameter, up to the presence of an extra discrete parameter. As argued above, this implies that the $tt$ and $rr$ components of the equations on SSS backgrounds \eqref{eq:sss} must be of first derivative order. Then, Proposition \ref{prop:sbprop} ensures that the theory is a type II QTG, thus showing the first item of the theorem.

Let us continue proving the second part of the theorem. Since the theory must be a type II QTG, the equations of motion on general SS backgrounds \eqref{eq:ss} take the form \eqref{eq:eqsss}, in terms of the function $h(\psi)$ that characterizes the theory. If static solutions $\mathcal{L}(g^{ab},R_{cdef})$ feature no other continuous parameters beyond the ADM mass, it follows that there does not exist $\psi_0 \in \mathbb{R}$ such that $h(\psi_0)=h'(\psi_0)=0$ --- cf. Remark \ref{rem:csp}. As a consequence, SS solutions must satisfy $\partial_t f=\partial_r N=0$, meaning that these must be static. Therefore, the theory fulfills Birkhoff's theorem. \qed

\begin{remark} Theorem \ref{thmBirk} implies that there exist no higher-curvature theories constructed from polynomial contractions of curvature invariants which satisfy a Birkhoff theorem and for which the corresponding most general SS solution involves two functions. As a matter of fact, the SSS solutions of most higher-curvature gravities are characterized by two functions --- see \eg \cite{Lu:2015cqa,Lu:2015psa}. Theorem \ref{thmBirk} implies that for all those theories such SSS solutions will not be the most general SS ones.
\end{remark}

Let us elaborate on this result. As explained before, it is easily shown that any theory satisfying a Birkhoff theorem must have second order equations on SSS configurations, so that it must be a type II QTG. Nevertheless, it is highly non-trivial that \emph{just} imposing SSS solutions to be characterized by a single continuous integration constant to be identified with the ADM mass --- featuring perhaps an additional discrete parameter --- actually enforces that all SS solutions are static. This key result is a consequence of Proposition \ref{prop:ssres}, which entails that possessing second-order equations on SSS backgrounds is equivalent to having second-order equations on general spherical configurations.

This discussion motivates the following question: given a QTG of type II, does it necessarily satisfy a Birkhoff theorem? The answer turns to be \emph{almost yes}, as there only exists a measure-zero set of QTGs of type II which do not fulfill Birkhoff's theorem because of a subtle technicality, as we explain in Remark \ref{rm:Birkpart} below. 

\begin{remark}
Let us go back to the expression for type II QTGs on spherical backgrounds, which was given in \eqref{eq:lagssform}. The equations of motion were given by \eqref{eq:eqsss} in terms of a function $h(\psi)$ specifying the theory. We observe that those theories for which there does not exist $\psi_0 \in \mathbb{R}$ such that $h(\psi_0)=h'(\psi_0)=0$ satisfy a Birkhoff theorem as stated in Definition \ref{def:birk}. 

Nevertheless, just like in Remark \ref{rem:csp}, those particular theories for which  $\exists  \, \psi_0 \in \mathbb{R}$ such that $h(\psi_0)=h'(\psi_0)=0$ will no longer satisfy a Birkhoff theorem, as any solution 
\begin{equation}
f(r)=1-\psi_0 r^2\,, \quad N(t,r)\,\,  \rm{arbitrary}
\end{equation}
solves the equation of motion.  Clearly, these theories are not generic and represent a lower-dimensional subspace of the moduli space of QTGs of type II.
\label{rm:Birkpart}
\end{remark}

In summary, we have proven that a theory constructed from polynomial curvature invariants satisfies Birkhoff's theorem if and only if it is a type II QTG whose static solutions are uniquely specified by a single continuous integration constant and a possible additional discrete parameter.

\section{QTGs and theories with $2^{\rm nd}$-order traced field equations}\label{sec6}

Let us now consider higher-curvature theories which possess second-order traced field equations on generic backgrounds. These satisfy Notion \ref{def:6} and correspond to QTGs of type III. These theories were considered in the literature before quasi-topological gravities of type I and II and may be regarded as their precursors \cite{Oliva:2010zd,Oliva:2010eb,Oliva:2011xu,Oliva:2012zs}. However, it has remained unclear to what extent these conditions are equivalent and/or imply each other. It is the purpose of this section to clarify this aspect.

To this end, it is convenient to analyze the condition a theory must satisfy for it to possess second-order traced field equations. Since the gravitational equations of motion for a higher-curvature theory $\mathcal{L}(g^{ab},R_{cdef})$ are given by \eqref{eq:eomgen}, we conclude that:
\begin{equation}
\text{Second-order traced field equations} \longleftrightarrow \nabla^c \nabla^d P_{cad}{}^a \, \text{is of second derivative order}\,.
\end{equation}
Let us restrict our attention to general SSS configurations \eqref{eq:sss}. Using \eqref{eq:nnPSSS}, direct computation reveals that:
\begin{align}
\notag
\left. \nabla^c \nabla^d P_{cad}{}^a \right \vert_{N,f}&=\frac{f}{2}\frac{\mathrm{d}\mathcal{C}_{N,f}}{\mathrm{d}r} + \frac{(rf'+(D-2)f)\mathcal{C}_{N,f}}{2r}
+\frac{f N' \mathcal{C}_{N,f}}{2N}\,, \\ &\mathcal{C}_{N,f}=A_{N,f}+(D-2)B_{N,f}\,,
\label{eq:nnPconform}
\end{align}
where $A_{N,f}$ and $B_{N,f}$ are presented in \eqref{eq:ABsss}. As a result, $\nabla^c \nabla^d P_{cad}{}^a$ will feature no derivatives of degree higher than two if and only if $A_{N,f}+(D-2)B_{N,f}$ is of first derivative order.

\begin{proposition}
\label{prop:2ndtr2}
Let $\mathcal{L}(g^{ab},R_{cdef})$ be a higher-curvature theory constructed from analytic combinations of polynomial curvature invariants. Assume it has second-order traced field equations. Then, it takes the following form:
\begin{equation}
\mathcal{L}(g^{ab},R_{cdef})=\mathcal{L}_1(g^{ab},W_{cdef})+\mathcal{L}_2(g^{ab},R_{cdef})\,,
\label{eq:2ndtr2}
\end{equation}
where $\mathcal{L}_1$ is written entirely in terms of Weyl curvature tensors and contractions thereof and $\mathcal{L}_2$ is a type II QTG and with second-order traced field equations on general backgrounds.  
\end{proposition}
\begin{proof}
By the structure of the curvature tensor \eqref{eq:RformNf} on a general SSS background \eqref{eq:sss}, the evaluation of the Lagrangian on \eqref{eq:sss},  $\mathcal{L}_{N,f}$,  will be entirely expressed in terms of the variables $\{\mathcal{R}_{N,f}^{(1)},\mathcal{R}_{N,f}^{(2)},\mathcal{R}_{N,f}^{(3)},\mathcal{R}_{N,f}^{(4)}\}$ defined in \eqref{eq:RformNfcomp}. Using \eqref{eq:pissss}, we may write:
\begin{align}\notag
\mathcal{C}_{N,f}&=2 \frac{\mathrm{d}}{\mathrm{d}r}\left (\frac{\partial  \mathcal{L}_{N,f}}{\partial \mathcal{R}_{N,f}^{(1)}}+\frac{\partial  \mathcal{L}_{N,f}}{\partial \mathcal{R}_{N,f}^{(3)}} \right)+\frac{2(D-2)}{r} \left (  \frac{\partial  \mathcal{L}_{N,f}}{\partial \mathcal{R}_{N,f}^{(1)}}+\frac{\partial  \mathcal{L}_{N,f}}{\partial \mathcal{R}_{N,f}^{(3)}}  \right )\\&+\frac{1}{2rf} \left[-4f\frac{\partial  \mathcal{L}_{N,f}}{\partial \mathcal{R}_{N,f}^{(4)}}+\frac{\partial  \mathcal{L}_{N,f}}{\partial \mathcal{R}_{N,f}^{(2)}}(r^2 \mathcal{R}_{N,f}^{(2)}-4f)- \frac{\partial  \mathcal{L}_{N,f}}{\partial \mathcal{R}_{N,f}^{(3)}}(4f+r^2 \mathcal{R}_{N,f}^{(2)}) \right]\,.
\end{align}
Clearly, the absence of third derivatives of $f$ and $N$ requires that:
\begin{equation}
\frac{\partial}{\partial \mathcal{R}_{N,f}^{(1)}} \left ( \frac{\partial  \mathcal{L}_{N,f}}{\partial \mathcal{R}_{N,f}^{(1)}}+\frac{\partial  \mathcal{L}_{N,f}}{\partial \mathcal{R}_{N,f}^{(3)}} \right)=0\,.
\label{eq:prueba2tra1}
\end{equation}
The structure of the curvature tensor implies that any curvature invariant will be symmetric under the exchange of $\mathcal{R}_{N,f}^{(2)}$ and $\mathcal{R}_{N,f}^{(3)}$. As a consequence, 
\begin{equation}
\frac{\partial}{\partial \mathcal{R}_{N,f}^{(1)}} \left ( \frac{\partial  \mathcal{L}_{N,f}}{\partial \mathcal{R}_{N,f}^{(1)}}+\frac{\partial  \mathcal{L}_{N,f}}{\partial \mathcal{R}_{N,f}^{(2)}} \right)=0\,.
\label{eq:prueba2tra2}
\end{equation}
Equations \eqref{eq:prueba2tra1} and \eqref{eq:prueba2tra2} enforce:
\begin{equation}
    \mathcal{L}_{N,f}=\mathcal{L}_{N,f}^{(1)}\left (\mathcal{R}_{N,f}^{(2)},\mathcal{R}_{N,f}^{(3)},\mathcal{R}_{N,f}^{(4)} \right)+\mathcal{L}_{N,f}^{(2)}\left (\mathcal{R}_{N,f}^{(1)}-\mathcal{R}_{N,f}^{(2)}-\mathcal{R}_{N,f}^{(3)},\mathcal{R}_{N,f}^{(4)} \right)
    \label{eq:lag2ndtraceprueba}
\end{equation}
where $\mathcal{L}_{N,f}^{(1)}$ and $\mathcal{L}_{N,f}^{(2)}$ are undetermined functions of the indicated variables. Now, $\mathcal{C}_{N,f}$ will feature no second derivatives of $f$ and $N$ if and only if:
\begin{align}
\frac{\partial^2 \mathcal{L}^{(2)}}{\partial x \partial \mathcal{R}_{N,f}^{(4)}}-2\frac{\partial^2 \mathcal{L}^{(2)}}{\partial x^2}- \frac{\partial^2 \mathcal{L}^{(1)}}{\partial \mathcal{R}_{N,f}^{(2)} \partial \mathcal{R}_{N,f}^{(3)}}&=0\,,   \\
\frac{\partial^2 \mathcal{L}^{(2)}}{\partial x \partial \mathcal{R}_{N,f}^{(4)}}-2\frac{\partial^2 \mathcal{L}^{(2)}}{\partial x^2}- \frac{\partial^2 \mathcal{L}^{(1)}}{\partial \mathcal{R}_{N,f}^{(2)} \partial \mathcal{R}_{N,f}^{(3)}}- \frac{\partial^2 \mathcal{L}^{(1)}}{ \partial \left ( \mathcal{R}_{N,f}^{(3)} \right)^2}&=0\,,
\end{align}
where $x=\mathcal{R}_{N,f}^{(1)}-\mathcal{R}_{N,f}^{(2)}-\mathcal{R}_{N,f}^{(3)}$. Finding the most general solution to these equations, taking into account the exchange symmetry between $\mathcal{R}_{N,f}^{(2)}$ and $\mathcal{R}_{N,f}^{(3)}$ and plugging the result into \eqref{eq:lag2ndtraceprueba}, one finds:
\begin{align} \notag
    \mathcal{L}_{N,f}&=\mathcal{H}_{N,f}^{(0)} \left ( \mathcal{R}_{N,f}^{(1)}-\mathcal{R}_{N,f}^{(2)}-\mathcal{R}_{N,f}^{(3)}+2\mathcal{R}_{N,f}^{(4)}\right)+ \mathcal{L}_{N,f}^{(3)}\left ( \mathcal{R}_{N,f}^{(4)} \right) \mathcal{R}_{N,f}^{(1)}  +\frac{\partial \mathcal{L}_{N,f}^{(3)}}{\partial \mathcal{R}_{N,f}^{(4)}}\mathcal{R}_{N,f}^{(2)} \mathcal{R}_{N,f}^{(3)}\\&+ \left ( \mathcal{R}_{N,f}^{(2)}+ \mathcal{R}_{N,f}^{(3)} \right) \mathcal{L}_{N,f}^{(4)}\left ( \mathcal{R}_{N,f}^{(4)} \right)+\mathcal{L}_{N,f}^{(5)}\left ( \mathcal{R}_{N,f}^{(4)} \right)\,,
\end{align}
where now $\mathcal{H}_{N,f}^{(0)},\mathcal{L}_{N,f}^{(3)},\mathcal{L}_{N,f}^{(4)}$ and $\mathcal{L}_{N,f}^{(5)}$ are single-variable functions of the indicated argument. On the one hand, the combination $ \mathcal{R}_{N,f}^{(1)}-\mathcal{R}_{N,f}^{(2)}-\mathcal{R}_{N,f}^{(3)}+2\mathcal{R}_{N,f}^{(4)}$ matches, up to a trivial overall constant factor, with the function $\Omega_{\rm SS}$ defined at \eqref{eq:omss} (evaluated on static configurations) that determines uniquely the Weyl curvature tensor on SSS backgrounds. As a consequence, we identify the piece with $\mathcal{H}_{N,f}^{(0)}$ as those coming from pure Weyl invariants. On the other hand, let us redefine:
\begin{equation}
   -2  \mathcal{L}_{N,f}^{(3)}\left (\mathcal{R}_{N,f}^{(4)}\right) \longleftrightarrow h_0(\psi)\,, \quad  -4  \mathcal{L}_{N,f}^{(4)}\left (\mathcal{R}_{N,f}^{(4)}\right) \longleftrightarrow h_1(\psi)\,, \quad \mathcal{L}_{N,f}^{(5)}\left ( \mathcal{R}_{N,f}^{(4)} \right)  \longleftrightarrow h_2(\psi)\,,
\end{equation}
where $\psi=\dfrac{1-f}{r^2}$. By Proposition \ref{prop2}, the functions $h_0, h_1$ and $h_2$ must take the form \eqref{eq:formhf} for the theory $\mathcal{L}$ to correspond to an analytic function of polynomial curvature invariants. Therefore, $\mathcal{L}_{N,f}-\mathcal{H}_{N,f}^{(0)} \left ( \mathcal{R}_{N,f}^{(1)}-\mathcal{R}_{N,f}^{(2)}-\mathcal{R}_{N,f}^{(3)}+2\mathcal{R}_{N,f}^{(4)}\right)$ must be a QTG of type II. Since $\mathcal{L}_{N,f}$ has second-order traced field equations by assumption, and pure Weyl invariants also possess second-order traced field equations \cite{Oliva:2010eb,Oliva:2011xu}, we conclude that $\mathcal{L}$ must be a type II QTG with second-order traced field equations modulo Weyl invariants.
\end{proof}

It is interesting to compare Proposition \ref{prop:2ndtr2} with the conjecture put forward in \cite{Oliva:2010zd}. There, it was claimed that any theory with second-order traced field equations is constructed, up to total derivatives, as a combination of Weyl invariants, theories of the Lovelock class and a very special theory of curvature order $k \geq 3$ existing in odd dimensions $D=2k-1$. Assuming the veracity of this conjecture, Proposition \ref{prop:2ndtr2} implies that such a special (non-topological) theory of order $k \geq 3$ in odd dimensions $D=2k-1$ belongs to the type II QTG class, thus providing an alternative proof for this fact \cite{Oliva:2011xu}.

Let us now make some comments regarding the reverse question: given a QTG of type II, will it feature second-order traced field equations? Interestingly enough, the answer is positive if one considers higher-curvature theories of at most cubic order in the curvature. Indeed, up to quadratic order in the curvature for $D \geq 5$, the unique theory fulfilling notions \ref{def:3}, \ref{def:4} and \ref{def:5} corresponds to Einstein-Gauss-Bonnet gravity, while at cubic order in the curvature it can be proven that the unique degeneracy in theories satisfying these notions comes from the density:
\begin{equation}
\mathcal{T}_{\rm W}=W_{abcd} W^{cdef} W_{ef}{}^{ab}-\frac{4(D^3-9D^2+26D-22)}{3D^2-15D+16} W_{abcd} W^{aecf} W_{e}{}^{b}{}_f{}^d\,, \quad D \geq 5,
\end{equation}
which vanishes identically on spherical backgrounds \eqref{eq:ss}. In $D=5$, this density is actually identically zero --- this is proven by expanding the identity $W_{[ab}{}^{cd} W_{cd}{}^{ef} W_{ef]}{}^{ab}=0$ ---, while it is non-trivial for $D \geq 6$. As a matter of fact, the cubic Lovelock density and the cubic quasi-topological theory discovered in \cite{Oliva:2010eb,Myers:2010ru} differ in $D \geq 6$ exactly by a multiple of this term. However, the inclusion of quartic terms changes the picture completely, as the following example shows.

\begin{remark}
\label{rem:cex2ndtrace}
Take the following combination of quartic densities in $D \geq 5$:
\begin{align}
\nonumber
\mathcal{T}_4&=\frac{(3D(D-4)+8)}{D-2} W_{abcd} W^{abcd} Z_{ef}Z^{ef}-\frac{4(D-4)(D-1)}{D-3} W_{abcd}Z^{ac} W^{bedf}Z_{ef}\\&-4D W_{abcd} W^{ebcd} Z^{af} Z_{ef}\,.
\end{align}
Direct computation shows that $\mathcal{T}_4$ vanishes identically for spherically symmetric backgrounds \eqref{eq:ss}, so that it is a trivial QTG of type II. However, one may compute:
\begin{align}
\nonumber
P^{ab}=g_{cd}\frac{\partial \mathcal{L}}{\partial R_{acbd}}=\frac{D-2}{2} &\left[\frac{(3D(D-4)+8)}{D-2} W_{cd}^{ef} W^{cd}_{ef} Z^{ab}-\frac{4(D-4)(D-1)}{D-3} W^{(a}{}_{c}{}^{b)d} W^{ec}_{fd} Z_{e}^{f}\right. \\&\left. -4D W^{(a}{}_{bcd} Z^{b)}{}_e W^{ebcd}+4 W_{cdef}W^{gdef} Z_g^c g^{ab} \right]\,,
\end{align}
and direct evaluation on generic non-spherical metric reveals that $\nabla^a \nabla^b P_{ab} \neq 0$, giving rise to higher derivatives in the traced field equations. As a result, QTGs of type II do not have in general second-order traced field equations. This issue may not be fixed by the addition of pure Weyl invariants, as these contribute at most with second-order derivatives to the traced field equations.

\end{remark}

{\noindent \bf Proof of Theorem \ref{thm4}.} Proposition \ref{prop:2ndtr2} implies that a QTG of type III may be converted into a QTG of type II through the addition of specific Weyl invariants. However, in spherical symmetry \eqref{eq:ss}, all Weyl invariants are proportional to the function $\Omega_{\rm SS}$ \cite{Deser:2005pc} defined in \eqref{eq:omss}. As a result, if one considers the Weyl invariants:
\begin{equation}
W_2=W_{abcd}W^{abcd}\,,\quad W_3=W_{abcd}W^{cdef}W_{ef}{}^{ab}\,,
\end{equation}
since these satisfy:
\begin{align}
    \left. W_2 \right \vert_{\rm SS}&=\frac{(D-1)(D-2)^2(D-3)}{4} \Omega^2_{\rm SS}\,, \quad \\     \left. W_3 \right \vert_{\rm SS}&=\frac{(D-1)(D-2)(D-3)(D^3-9D^2+26D-22)}{8} \Omega^3_{\rm SS}\,,
\end{align}
then any Weyl invariant of order $n$ in the curvature, denoted schematically as $\mathcal{W}^{(n)}$, can be expressed as:
\begin{equation}
    \left. \mathcal{W}^{(2n)} \right \vert_{\rm SS}=\alpha_{2n} \left. \left ( W_2\right)^n \right \vert_{\rm SS}\,, \quad \left. \mathcal{W}^{(2n+1)} \right \vert_{\rm SS}=\alpha_{2n+1} \left. \left ( W_2\right)^{n-1} \right \vert_{\rm SS}  \left. W_3 \right \vert_{\rm SS}\,, \quad n \geq 1\,,
\end{equation}
for some constants $\alpha_{2n}$ and $\alpha_{2n+1}$. Therefore, one just needs to add an analytic function of $W_2$ and $W_3$ to a QTG of type III to make it a QTG of type II and we conclude. \qed

\section{Future directions}

The main results of the manuscript were summarized in the introduction. Here we list a few directions in which this work could be extended.

\paragraph{Alternate notions of Birkhoff theorem.} How does one define Birkhoff theorem beyond GR? There are other possible definitions one could consider, in addition to the one we have worked with here --- see Definition \ref{def:birk}. For example: (1) The imposition of spherical symmetry leads necessarily to the existence of an additional, hypersurface-orthogonal time-like Killing vector field --- \ie all spherically symmetric solutions are static. (2) Any spherically symmetric solution is locally isometric to \textit{the unique} static and spherically symmetric solution. (3) The most general SS solution of the equations of motion is static up to transformations that are symmetries of the theory (e.g.~conformal transformations). Restricted to vacuum GR, all these notions of a Birkhoff theorem coincide, but this is not obviously the case in beyond GR theories. For example, in pure Weyl$^2$ gravity~\cite{Riegert:1984zz} the most general SS solution is static up to a conformal transformation, thus satisfying (3) but not (1) or (2). On the other hand, in generic polynomial type II QTGs, the SSS solution has a discrete nonuniqueness ($f(r)$ is the solution of a polynomial equation) and hence does not satisfy (2). Similarly, if the equations of motion on SSS backgrounds are of higher derivative order, (1) could be fulfilled but not our Definition \ref{def:birk}, as justified in the main body of the document. Understanding the actual \textit{space of Birkhoff theorems} in beyond-GR theories is therefore in itself a question worth exploring. 

From another perspective, it is often claimed that Birkhoff's theorem is related to the absence of propagating spin-0 modes in the theory. For example, this is conjectured in~\cite{Riegert:1984zz} based on the spectrum of Weyl$^2$ gravity around flat spacetime. We are not aware of any proof of the validity of this claim (though see~\cite{Deser:2006mz}). If Birkhoff's theorem can be related to the spin-0 degrees of freedom of the theory in question, it cannot be the spectrum around the vacuum that is relevant but instead the spectrum on the spherical black hole background. It would be interesting to study the implications of enforcing the absence of scalar modes for perturbations of spherically symmetric spacetimes in general theories of gravity to see whether these conditions are the same as or different from the defining properties of type II QTGs.

\paragraph{Theories with non-minimally coupled matter.} An interesting class of theories to which our analysis may be extended is electromagnetic quasi-topological gravities (EMQTGs)~\cite{Cano:2020qhy,Cano:2020ezi,Bueno:2021krl,Cano:2022ord,Bueno:2022ewf}. These are theories of gravity non-minimally coupled to a $U(1)$ gauge field for which the defining characteristic is that the equations of motion on SSS metrics are second-order. Hence, the known EMQTGs are non-minimal extensions of type II QTGs. There remains a considerable space of exploration for EMQTGs: do different notions of EMQT theories exist? Do any such theories satisfy a Birkhoff theorem? Performing a more rigorous classification of such theories and their properties would be a natural aim for future work.

\paragraph{Generalized quasi-topological theories.}
Type I QTGs are in fact a subset of an even broader class of theories known as ``Generalized quasi-topological gravities'' (GQTGs) \cite{Hennigar:2017ego,  Bueno:2019ycr,Bueno:2022res,Moreno:2023arp,Moreno:2023rfl}.
Unlike QTGs,  GQTGs built from polynomial densities exist at every curvature order in $D=4 $ \cite{Bueno:2016xff,Hennigar:2016gkm, Bueno:2017qce,Bueno:2017sui}. GQTGs are usually defined through a slight modification of Notion \ref{def:1}. Indeed, a higher-curvature theory with Lagrangian density $\mathcal{L}(g^{ab}, R_{cdef})$ is a GQTG if
\begin{equation}
\frac{\delta L_{f}}{\delta f}=0 \, ,
\label{eq:defgqtg}
\end{equation}
namely, if the variation of the on-shell Lagrangian on the single-function SSS ansatz is identically zero. If this definition holds, it can be shown that the theory admits SSS solutions characterized by a single function $f(r)$. The equation for $f(r)$ can always be integrated once, giving rise to an equation which is either algebraic (QTG case) or second-order in derivatives of $f(r)$ and linear in $f''(r)$ --- see \cite{Bueno:2017sui}. As argued in the same reference, the above definition implies a modified version of Notion \ref{def:3}, for which \req{defi3} is replaced by an analogous expression where now $\mathcal{F}$ may depend on derivatives of $f$. It is natural to wonder if alternative equivalent/broader/narrower notions of GQTGs exist. In particular, it would be interesting to explore GQTG variations of Notions \ref{def:2} and \ref{def:4}: could \eqref{eq:defgqtg} be equivalently formulated in terms of some conditions on $\nabla_a P^{abcd}$ when considered on \eqref{eq:sssn1} and/or \eqref{eq:sss}? For example, since \eqref{eq:defgqtg} implies that the $tt$ and $rr$ components of the equations of motion must satisfy $\mathcal{E}_t^t=\mathcal{E}_r^r$ --- cf. \eqref{eq:derfuncond}---, this would impose that the $tt$ component of $\nabla_a\nabla_c P^{abcd}$ on \eqref{eq:sssn1} must be of third derivative order.

\paragraph{Extension to non-polynomial Lagrangians.} It has been recently realized that there exist theories constructed from non-polynomial curvature invariants that satisfy the  various defining characteristics of type II quasi-topological gravities~\cite{Bueno:2025zaj} (see also~\cite{Colleaux:2019ckh}). We have not focused here on this class of theories, and indeed some of our proof techniques relied on the analytic form of the Lagrangian --- see the proof of Proposition \ref{prop2}. In light of this, it would be particularly interesting to perform a similar analysis as done here for theories of a non-polynomial nature. One could classify how the different notions of QTGs apply in that context, and which of the results presented here go through in the non-polynomial case. In such an analysis it will likely turn out that there exist theories which have second-order equations of motion that differ from the field equations of polynomial QTGs. Indeed, this has recently been argued in~\cite{Carballo-Rubio:2025ntd} by `lifting' the equations of motion of two-dimensional Horndeski theories to higher-dimensional spherically symmetric spacetimes. In this vein, a natural question concerns whether there are any restrictions on the form of such theories that can arise from genuine, covariant higher-dimensional theories. That is: for each choice of two-dimensional Horndeski theory, does there exist a corresponding higher-dimensional theory that produces it when spherically reduced?

\subsection*{Acknowledgements} We wish to thank Pablo A. Cano and  Alejandro Vilar L\'opez for useful discussions and early collaboration on related topics.
PB was supported by a Ram\'on y Cajal fellowship (RYC2020-028756-I), by a Proyecto de Consolidación Investigadora (CNS 2023-143822) from Spain’s Ministry of Science, Innovation and Universities, and by the grant PID2022-136224NB-C22, funded by MCIN/AEI/ 10.13039/501100011033/FEDER, UE.
  \'AJM  was supported by a Juan de la Cierva contract (JDC2023-050770-I) from Spain’s Ministry of Science, Innovation and Universities.

  \appendix
  \section{An alternative characterization for type II QTGs}

In this appendix, we give an equivalent characterization for type II QTGs in terms of the structure of the off-diagonal components of the equations of motion on general spherical backgrounds \eqref{eq:ss}.

To set the stage, consider a higher-curvature theory $\mathcal{L}(g^{ab},R_{cdef})$ constructed from analytic functions of polynomial curvature invariants. Assume that the $tr$ component of the equations on spherical backgrounds \eqref{eq:ss} with $N=N(r)$
possess the structure:
    \begin{equation}
        \left. \mathcal{E}_{tr} \right \vert_{{\rm SS}, N=N(r)}=\mathcal{F}\left (r,f,N, \partial_t f \right )\,.
        \label{eq:etrformsplit}
    \end{equation}
We will say that such a theory $\mathcal{L}(g^{ab},R_{cdef})$ satisfies a \emph{temporal splitting property}.

Let us explore the conditions resulting from a theory satisfying the temporal splitting property.  To this aim, it is convenient to compute $\left. \nabla^a \nabla_c P_{ab}{}^{cd} \right \vert_{\rm SS}$ for a generic theory  $\mathcal{L}(g^{ab},R_{cdef})$ --- not necessarily fulfilling \eqref{eq:etrformsplit} ---  on our general SS ansatz \eqref{eq:ss}:
\begin{align}
\notag
\left. \nabla^a \nabla_c P_{ab}{}^{cd} \right \vert_{\rm SS}&=\left (\frac{f \partial_r B_{\rm SS}}{2}+\frac{B_{\rm SS} \partial_r f}{4}+\frac{(D-2)f B_{\rm SS}}{2r}+\frac{A_{\rm SS}\partial_t f}{4N^2f^2} \right) \tau_b^d\\\notag  & \hspace{-1cm}+\left ( -\frac{\partial_t A_{\rm SS}}{2 N^2 f}+\frac{A_{\rm SS}(2f\partial_t N+N \partial_t f)}{4f^2 N^3}+\frac{B_{\rm SS}(2\partial_r N f+ N\partial_r f)}{4N}+\frac{(D-2)f E_{\rm SS}}{2r} \right) \rho_b^d\\  \notag & \hspace{-1cm}+\left (\frac{f\partial_r E_{\rm SS}}{2}-\frac{\partial_t C_{\rm SS}}{2 N^2f}+\frac{C_{\rm SS}(N \partial_t f+f \partial_t N)}{2N^3f^2}+\frac{E_{\rm SS}(f\partial_r N+N \partial_r f)}{2N}+\frac{(D-3) f E_{\rm SS}}{2r} \right) \sigma_b^d \\ \label{eq:nablanablaPss} & \hspace{-1cm}+ \left (\frac{A_{\rm SS}(2\partial_r N f+\partial_r f N)}{4f N^2}-\frac{B_{\rm SS} \partial_t f}{4fN}-\frac{\partial_t B_{\rm SS}}{2N} \right) \zeta_b{}^d\,, 
\end{align}
where in the derivation of the result we have used the property \eqref{eq:propcuriosa}.

\begin{proposition}
Let $\mathcal{L}(g^{ab},R_{cdef})$ be a higher-curvature gravity constructed from polynomial contractions of curvature invariants. If it satisfies the temporal splitting property, then the theory is a type II QTG.
\label{prop:altqtg2}
\end{proposition}
\begin{proof}
We will be working with spherical backgrounds \eqref{eq:ss} with $N=N(r)$. We will add the subscript $\rm{rSS}$ to every object evaluated on this restricted SS ansatz, or just change the label $\rm{SS}$ into the new label $\rm{rSS}$ to remark that those quantities must be considered in the particular case $N=N(r)$.

From the structure of the curvature tensor \eqref{eq:Rformss} and the argumentation in the proof of Proposition \ref{prop:ssfoem}, the Lagrangian on the spherically symmetric background \eqref{eq:ss} with $N=N(r)$ will be a function:
\begin{equation}
    \mathcal{L}_{\rm rSS}= \mathcal{L}_{\rm rSS}\left ( \mathcal{R}_{\rm rSS}^{(1)}, \mathcal{R}_{\rm rSS}^{(2)}+ \mathcal{R}_{\rm rSS}^{(3)}, \left ( \mathcal{R}_{\rm rSS}^{(2)}- \mathcal{R}_{\rm rSS}^{(3)} \right)^2 -4 \left ( \mathcal{R}_{\rm rSS}^{(4)} \right)^2, \mathcal{R}_{\rm rSS}^{(5)}\right)\,.
\end{equation}
Focus on the $tr$ component $\mathcal{E}_{tr}$ of the equations of motion. For it to have the form \eqref{eq:etrformsplit}, it cannot contain any fourth derivatives of $f$. In particular, no $\partial_r^3 \partial_t f$ may appear in $\mathcal{E}_{tr}$. By \eqref{eq:Rformss}, \eqref{eq:nablanablaPss} and \eqref{eq:ABss}, the $tr$ component of the equations will be safe from the appearance of this term if and only if:
\begin{equation}
   \frac{\partial \left( \partial_t B_{\rm rSS} \right)}{\partial \left (\partial_r^3 \partial_t f \right)}=0 \rightarrow  \frac{\partial^2 \mathcal{L}_{\rm rSS}}{\partial \left (\mathcal{R}_{\rm rSS}^{(1)} \right)^2}=0\,.
   \label{eq:R1linearb}
\end{equation}
From here, we directly infer that:
\begin{equation}
 \mathcal{L}_{\rm rSS}=  \mathcal{L}_{\rm rSS}^{(0)}\left (\mathfrak{u},\mathfrak{v}, \mathcal{R}_{\rm rSS}^{(5)}\right) \mathcal{R}_{\rm rSS}^{(1)} +\mathcal{L}_{\rm rSS}^{(1)}\left (\mathfrak{u},\mathfrak{v}, \mathcal{R}_{\rm rSS}^{(5)}\right)\,,
\end{equation}
where we have defined $\mathfrak{u}=\mathcal{R}_{\rm rSS}^{(2)}+ \mathcal{R}_{\rm rSS}^{(3)}$, $\mathfrak{v}= \left ( \mathcal{R}_{\rm rSS}^{(2)}- \mathcal{R}_{\rm rSS}^{(3)} \right)^2 -4 \left ( \mathcal{R}_{\rm rSS}^{(4)} \right)^2$ and where $ \mathcal{L}_{\rm rSS}^{(0)}$ and $ \mathcal{L}_{\rm rSS}^{(1)}$ are certain functions of the indicated variables. As it turns out, $\mathcal{E}_{tr}$ is now free from fourth derivatives.

Let us now explore the presence of pure third  derivatives in $\mathcal{E}_{tr}$. Specifically, third derivatives of the form $\partial_r^2 \partial_t f$ must be absent. As it turns out, these may only arise from the piece $\partial_ t B_{\rm rSS}$ in \eqref{eq:nablanablaPss}, so we need to impose:
\begin{equation}
    \frac{\partial \left ( \partial_t B_{\rm rSS} \right) }{\partial \left ( \partial_r ^2 \partial_t f \right)}=0 \rightarrow  \frac{\partial^2 \mathcal{L}_{\rm rSS}}{\partial \mathfrak{u} \, \partial \mathcal{R}_{\rm rSS}^{(1)}}-2 \left (\mathcal{R}_{\rm rSS}^{(2)}-\mathcal{R}_{\rm rSS}^{(3)} \right)  \frac{\partial^2 \mathcal{L}_{\rm rSS}}{\partial \mathfrak{v} \, \partial \mathcal{R}_{\rm rSS}^{(1)}}=0\,.
\end{equation}
Since $\left (\mathcal{R}_{\rm rSS}^{(2)}-\mathcal{R}_{\rm rSS}^{(3)} \right)$ cannot be solely expressed in terms of $\mathfrak{u}$ and $\mathfrak{v}$, we learn that:
\begin{equation}
     \frac{\partial^2 \mathcal{L}_{\rm rSS}}{\partial \mathfrak{u} \, \partial \mathcal{R}_{\rm rSS}^{(1)}}= \frac{\partial^2 \mathcal{L}_{\rm rSS}}{\partial \mathfrak{v} \, \partial \mathcal{R}_{\rm rSS}^{(1)}}= \frac{\partial \mathcal{L}_{\rm rSS}^{(0)}}{\partial \mathfrak{u}}=\frac{\partial \mathcal{L}_{\rm rSS}^{(0)}}{\partial \mathfrak{v}}=0\,.
\end{equation}
As a consequence:
\begin{equation}
     \mathcal{L}_{\rm rSS}=  \mathcal{L}_{\rm rSS}^{(0)}\left (\mathcal{R}_{\rm rSS}^{(5)}\right) \mathcal{R}_{\rm rSS}^{(1)} +\mathcal{L}_{\rm rSS}^{(1)}\left ( \mathfrak{u}, \mathfrak{v}, \mathcal{R}_{\rm rSS}^{(5)}\right)\,.
     \label{eq:lagssbirkporahora}
\end{equation}
After imposing \eqref{eq:lagssbirkporahora}, no third derivatives appear in $\mathcal{E}_{tr}$. Let us now study the possible presence of second derivatives. Let us concentrate on the appearance of mixed derivatives $\partial_t \partial_r f$, which may only come from the term $\left. \nabla^a \nabla_c P_{ab}{}^{cd} \right \vert_{\rm rSS}$ --- just have a look at the form of the curvature tensor \eqref{eq:Rformss}, which is free of such mixed derivatives. Since $A_{\rm rSS}$ and $B_{\rm rSS}$ are now of first derivative order, mixed derivatives may only arise from $\partial_t B_{\rm rSS}$:
\begin{align}
\frac{\partial \left ( \partial_t B_{\rm rSS} \right) }{\partial \left ( \partial_t \partial_r f\right)}=0 \rightarrow \quad \frac{\partial \mathcal{L}_{\rm rSS}^{(0)} }{\partial \mathcal{R}_{\rm rSS}^{(5)}}-2 \frac{\partial^2  \mathcal{L}_{\rm rSS}^{(1)}}{\partial \mathfrak{u}^2}-4\left ( \mathcal{R}_{\rm rSS}^{(2)} -\mathcal{R}_{\rm rSS}^{(3)} \right)\frac{\partial^2  \mathcal{L}_{\rm rSS}^{(1)}}{\partial \mathfrak{u} \partial \mathfrak{v}}=0\,.
\end{align}
Since $\mathcal{L}_{\rm rSS}^{(0)}$ depends solely on $\mathcal{R}_{\rm rSS}^{(5)}$ and $\mathcal{L}_{\rm rSS}^{(1)}$ depends on $\mathfrak{u}, \mathfrak{v}$ and $\mathcal{R}_{\rm rSS}^{(5)}$, we conclude that:
\begin{equation}
    \frac{\partial^2 \mathcal{L}_{\rm rSS}^{(1)}}{\partial \mathfrak{v} \partial \mathfrak{u}}=0\,, \quad  \frac{\partial \mathcal{L}_{\rm rSS}^{(0)} }{\partial \mathcal{R}_{\rm rSS}^{(5)}}-2 \frac{\partial^2  \mathcal{L}_{\rm rSS}^{(1)}}{\partial \mathfrak{u}^2}=0\,.
\end{equation}
Solving these partial differential equations, one arrives to:
\begin{align}
    \mathcal{L}_{\rm rSS}=\mathcal{L}_{\rm rSS}^{(0)}\left (\mathcal{R}_{\rm rSS}^{(5)}\right) \mathcal{R}_{\rm rSS}^{(1)} +\frac{\mathfrak{u}^2}{4}\frac{\partial \mathcal{L}_{\rm rSS}^{(0)} }{\partial \mathcal{R}_{\rm rSS}^{(5)}} + \mathfrak{u}\, \mathcal{L}_{\rm rSS}^{(2)}\left ( \mathcal{R}_{\rm rSS}^{(5)} \right) 
    +\mathcal{L}_{\rm rSS}^{(3)}\left ( \mathfrak{v},\mathcal{R}_{\rm rSS}^{(5)} \right) \,.
    \label{eq:ultimopasobirk}
\end{align}
where $\mathcal{L}_{\rm rSS}^{(2)}$ is a certain function of $\mathcal{R}_{\rm rSS}^{(5)}$ and $\mathcal{L}_{\rm rSS}^{(3)}$ is an unspecified function of $\mathfrak{v}$ and $\mathcal{R}_{\rm rSS}^{(5)}$. 

Let us also examine now the presence of second temporal derivatives $\partial_t^2 f$ in $\mathcal{E}_{tr}$. From the structure of \eqref{eq:ultimopasobirk}, it is easy to check that those second derivatives may only be due to the term $\partial_t B_{\rm rSS}$ in $\left. \nabla^a \nabla_c P_{ab}{}^{cd} \right \vert_{\rm rSS}$. We get:
\begin{equation}
    \frac{\partial \left ( \partial_t B_{\rm rSS} \right) }{\partial \left ( \partial_t^2 f\right)} \rightarrow \quad \frac{\partial^2  \mathcal{L}_{\rm rSS}^{(3)}}{\partial \mathfrak{v}^2}=0\,,
\end{equation}
so \eqref{eq:ultimopasobirk} boils down to:
\begin{align}
    \mathcal{L}_{\rm rSS}=\mathcal{L}_{\rm rSS}^{(0)}\left (\mathcal{R}_{\rm rSS}^{(5)}\right) \mathcal{R}_{\rm rSS}^{(1)} +\frac{\mathfrak{u}^2}{4}\frac{\partial \mathcal{L}_{\rm rSS}^{(0)} }{\partial \mathcal{R}_{\rm rSS}^{(5)}} + \mathfrak{u}\, \mathcal{L}_{\rm rSS}^{(2)}\left ( \mathcal{R}_{\rm rSS}^{(5)} \right) 
    +\mathfrak{v}\mathcal{L}_{\rm rSS}^{(3)}\left (\mathcal{R}_{\rm rSS}^{(5)} \right)+ \mathcal{L}_{\rm rSS}^{(4)}\left (\mathcal{R}_{\rm rSS}^{(5)} \right)\,,
    \label{eq:ultimopasobirk2}
\end{align}
where $\mathcal{L}_{\rm rSS}^{(4)}$ is an unspecified function of $\mathcal{R}_{\rm rSS}^{(5)}$.

After having imposed this, no second derivatives appear in $\mathcal{E}_{tr}$. Therefore, we just need to ensure the absence of radial derivatives. In particular, let us analyze the presence of derivatives $\partial_r N=N'$. Direct computation shows:
\begin{equation}
     \frac{\partial}{\partial \left ( \partial N' \right)} \left (\left. P_{te}{}^{cd} R_{cd}{}^{re}\right \vert_{\rm rSS}+2 \left. \nabla^a \nabla_c P_{ab}{}^{cd} \right \vert_{\rm rSS}\right)=0 \rightarrow \quad  \mathcal{L}_{\rm rSS}^{(3)}=-\frac{1}{4}\frac{\partial \mathcal{L}_{\rm rSS}^{(0)} }{\partial \mathcal{R}_{\rm rSS}^{(5)}}\,.
\end{equation}
As a consequence, we conclude that:
\begin{align}
    \mathcal{L}_{\rm rSS}&= \mathcal{L}_{\rm rSS}^{(0)}\left (\mathcal{R}_{\rm rSS}^{(5)}\right)\mathcal{R}_{\rm rSS}^{(1)} +\frac{\left (\mathfrak{u}^2 -\mathfrak{v}\right)}{4} \frac{\partial \mathcal{L}_{\rm rSS}^{(0)}}{\partial \mathcal{R}_{\rm rSS}^{(5)}} + \mathfrak{u}\, \mathcal{L}_{\rm rSS}^{(2)}\left ( \mathcal{R}_{\rm rSS}^{(5)} \right) +\mathcal{L}_{\rm rSS}^{(4)}\left ( \mathcal{R}_{\rm rSS}^{(5)} \right) \,.
\end{align}
Finally, let us now evaluate $\mathcal{L}_{\rm rSS}$ on the single function SSS ansatz with $N=1$. By making the redefinitions:
\begin{equation}
   -2  \mathcal{L}_{\rm rSS}^{(0)}\left (\mathcal{R}_{\rm rSS}^{(5)}\right) \longleftrightarrow h_0(\psi)\,, \quad  -4  \mathcal{L}_{\rm rSS}^{(2)}\left (\mathcal{R}_{\rm rSS}^{(5)}\right) \longleftrightarrow h_1(\psi)\,, \quad \mathcal{L}_{\rm rSS}^{(4)}\left ( \mathcal{R}_{\rm rSS}^{(5)} \right)  \longleftrightarrow h_2(\psi)\,,
\end{equation}
where $r^2 \psi=1-f$, we observe that the theory on \eqref{eq:sssn1} would take precisely the form \eqref{eq:laghf}. However, Proposition \ref{prop2} fixes the functions $h_0,h_1$ and $h_2$ to have the form \eqref{eq:formhf}. Therefore, $\mathcal{L}_{\rm rSS}$ must have the form \eqref{eq:lagssform} when evaluated on backgrounds $N=N(r)$. However, if we restore the time dependence on $N$, since the only appearance of temporal derivatives of $N$ in the curvature tensor is through $\mathcal{R}_{\rm rSS}^{(1)}$, we see that the Lagrangian on general spherical backgrounds \eqref{eq:ss} must take the form \eqref{eq:lagssform}. Since this corresponds to the Lagrangian of  a type II QTG on \eqref{eq:ss} with $N=N(r)$, we conclude.

\end{proof}

As a consequence, Proposition \ref{prop:altqtg2} proves that the temporal splitting property provides an equivalent characterization for type II QTGs .

\bibliographystyle{JHEP-2}
\bibliography{Gravities.bib}

\end{document}